\documentclass[journal]{IEEEtran}

\usepackage{amsmath}
\usepackage{bm}
\usepackage{amssymb}
\usepackage{graphicx}
\usepackage{color, colortbl} 
\usepackage[dvipsnames,svgnames,x11names]{xcolor}
\usepackage{algpseudocode,algorithm}
\usepackage{hyperref}
\usepackage{enumerate}
\usepackage[caption=false,font=footnotesize]{subfig}
\usepackage{cite}
\usepackage{url,comment}
\usepackage[normalem]{ulem}
\usepackage[shortlabels]{enumitem}
\usepackage{amsthm}
\usepackage{balance}
\usepackage{etoolbox}
\usepackage{setspace}

\makeatletter
\patchcmd{\@makecaption}
  {\scshape}
  {}
  {}
  {}
\makeatother

\newcommand{\appropto}{\mathrel{\vcenter{
  \offinterlineskip\halign{\hfil$##$\cr
    \propto\cr\noalign{\kern2pt}\sim\cr\noalign{\kern-2pt}}}}}

\newtheorem{lemma}{Lemma}
\newtheorem{prop}{Proposition}
\newtheorem{rem}{Remark}

\begin{document}
\title{PMBM-based SLAM Filters in 5G mmWave Vehicular Networks}
\author{
\IEEEauthorblockN{Hyowon Kim, \IEEEmembership{Member, IEEE}, Karl Granstr\"{o}m, \IEEEmembership{Member, IEEE}, Lennart Svensson, \IEEEmembership{Senior Member, IEEE}, \\
Sunwoo Kim, \IEEEmembership{Senior Member, IEEE}, and Henk Wymeersch, \IEEEmembership{Senior Member, IEEE}}
\thanks{This paper is accepted for publication in IEEE Transactions on Vehicular Technology.}
\thanks{Copyright (c) 2015 IEEE. Personal use of this material is permitted. However, permission to use this material for any other purposes must be obtained from the IEEE by sending a request to pubs-permissions@ieee.org.}
\thanks{H. Kim, L. Svensson and H. Wymeersch are with the Department of Electrical Engineering, Chalmers University of Technology, 412 58 Gothenburg, Sweden (email: hyowon@chalmers.se; lennart.svensson@chalmers.se; henkw@chalmers.se).}
\thanks{K. Granstr\"{o}m is with Embark Trucks Inc., San Francisco, CA 94107, USA (email: kagranstrom@gmail.com).}
\thanks{S. Kim is with the Department of Electronic Engineering, Hanyang University, 04763 Seoul, South Korea (email: remero@hanyang.ac.kr).}

}

\maketitle
\begin{abstract}
    Radio-based vehicular simultaneous localization and mapping (SLAM) aims to localize vehicles while mapping the landmarks in the environment.
    We propose a sequence of three Poisson multi-Bernoulli mixture~(PMBM) based SLAM filters, which handle the entire SLAM problem in a theoretically optimal manner.
    The complexity of the three proposed SLAM filters is progressively reduced while sustaining high accuracy by deriving SLAM density approximation with the marginalization of nuisance parameters (either vehicle state or data association).
    Firstly, the PMBM SLAM filter serves as the foundation, for which we provide the first complete description based on a Rao-Blackwellized particle filter. Secondly, the Poisson multi-Bernoulli (PMB) SLAM filter is based on the standard reduction from PMBM to PMB, but involves a novel interpretation based on auxiliary variables and a relation to Bethe free energy. Finally, using the same auxiliary variable argument, we derive a marginalized PMB SLAM filter, which avoids particles and is instead implemented with a low-complexity cubature Kalman filter.
    We evaluate the three proposed SLAM filters in comparison with the probability hypothesis density~(PHD) SLAM filter in 5G mmWave vehicular networks and show the computation-performance trade-off between them.
\end{abstract}


\IEEEpeerreviewmaketitle
\vspace{-2mm}
\section{Introduction} \label{sec:Introduction}

    In 5G vehicular networks, mmWave signals with large bandwidths bring high resolution in both time-delay and angle domains~\cite{WymSecDesDarTuf:J18}.
    This makes it possible for a 5G mmWave receiver on a vehicle to perform simultaneous localization and mapping (SLAM), i.e., to both exploit the multipath for improving positioning and for using position information to map the environment, which we define as 5G radio-SLAM~(see Fig.~\ref{Fig:scenario})~\cite{Witrisal_SPM2016,Erik_BPSLAM_TWC2019,Hyowon_TWC2020}.
    %
    {The SLAM problem~\cite{Durrant2006SLAM1,Durrant2006SLAM2} is in general divided into a front-end and a back-end problem.
    The front-end problem is to determine the association between landmarks and measurement detections, known as data association and is highly dependent on sensor type.
    The back-end problem is to find the probabilistic SLAM density given data association determined in the front-end problem.
    In the 5G mmWave radio SLAM scenario, this data association is challenging due to the following aspects:
    i) lack of features that identify which landmark generated the corresponding detection;
    ii) high dimensional vehicle state with uncertainty;
    iii) missed detections at the receiver, due to errors in the detection process or varying sensor field-of-view (FoV); 
    iv) false detections, due to errors in measurement routine or clutter.
    Even with known data associations, the back-end SLAM problem is still challenging due to the coupling between the unknown vehicle state trajectory and the unknown landmark states.}
    To solve the SLAM problem, several approaches have been developed, and
    these are now described in detail.

\begin{figure}
\begin{centering}
	\includegraphics[width=.8\columnwidth]{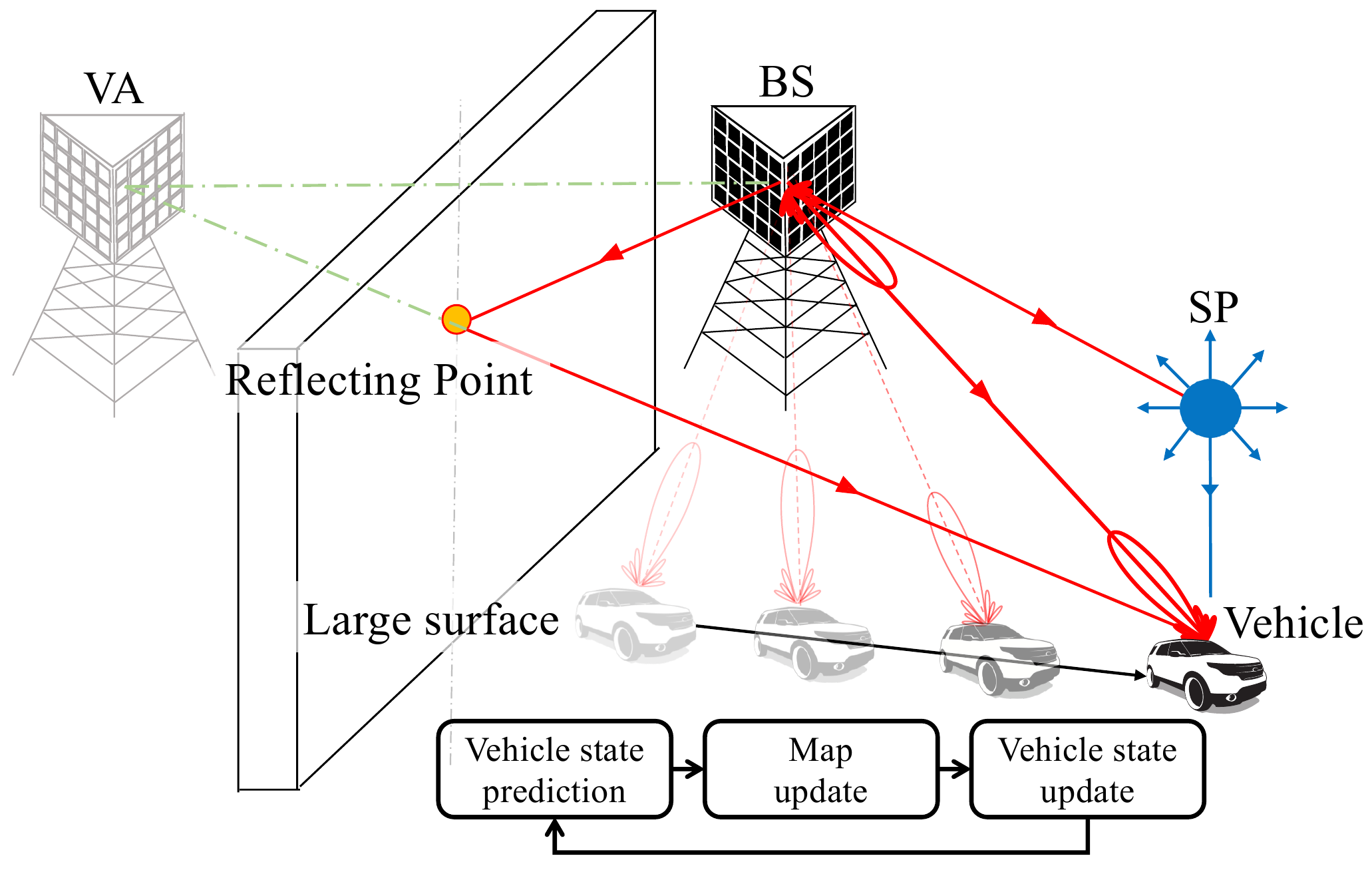}
	\caption{Example of 5G radio-SLAM application. The base station~(BS) transmits the mmWave signals, and multipath components are generated depending on the propagation environment: i) reflected by the large surfaces which are characterizing the virtual anchors (VAs); and ii) scattered by the scattering points (SPs). {A high-level flow-chart of SLAM is shown.}}
	\label{Fig:scenario}
\end{centering}
  \vspace{-6mm}
\end{figure}

    The most well-known SLAM algorithms such as extended Kalman filter~(EKF) SLAM~\cite{Randall_EKFSLAM,Gamini_EKFSLAM}, FastSLAM~\cite{FastSLAM,FastSLAM2,Gentner_SLAM_TWC2016}, and GraphSLAM~\cite{GraphSLAM_2006,GraphSLAM} are Bayesian methods, based on random vectors, to solve on the back-end problem.
    In EKF SLAM, the sensor and landmark states are collected in a single vector, and the posterior is approximated as Gaussian using an EKF. 
    In FastSLAM, the sensor and landmark states are estimated using a Rao-Blackwellized particle filter~(RBPF) where the 
    the sensor state posterior is handled using a particle filter. 
    GraphSLAM~\cite{GraphSLAM_2006,GraphSLAM} makes use of a graphical model to efficiently solve for the maximum a posteriori estimates of the sensor state trajectory and the landmark states by constrained optimization. 
    However, those data associations are explicitly modeled out of Bayesian SLAM filter, and false alarms cannot be properly handled. Instead, the front-end problem (determining the data association) is handled by separate approaches such as the Mahalanobis distance test~\cite{Gamini_EKFSLAM}, maximum likelihood test~\cite{FastSLAM2}, and joint compatibility test~\cite{JCtestDA_TRA2001}.


    
    
    An alternative approach is based on random finite sets~(RFSs)~\cite{mahler_AES_2003_PHD,mahler_book_2007} rather than random vectors, leading to
    a rigorous and powerful framework for solving SLAM problems.  RFS-based SLAM naturally captures the data association problem, the randomness in the number of measurements and number of detected landmarks, and can address both front-end and back-end problems 
    in a fully Bayesian manner. 
    Several RFS-based SLAM filters have been proposed in the literature. 
    {The probability hypothesis density~(PHD) SLAM filter~\cite{Mullane2011} is proposed from~\cite{Vo2006}, implemented by the RBPF.
    With the measurement at the current time, the proposal distribution for sampling the particles is computed~\cite{Lin_PHD-SLAM2.0_TR2021}, which allows for the reduced number of particles.
    Furthermore, the Poisson RFS-likelihood for computing the vehicle posterior density is theoretically derived in~\cite{Hyowon_TWC2020}.
    However, the PHD-SLAM filter is vulnerable to missed detections and false alarms since it cannot track the landmarks without explicit data association.}
    {The labeled multi-Bernoulli~(LMB) SLAM filter~\cite{Deusch2015,DeuschRD:2015} and $\delta$-generalized labeled multi-Bernoulli~($\delta$-GLMB) SLAM filter~\cite{Diluka_GLMB-SLAM_ICCAIS2018,Diluka_GLMB-SLAM_Sensor2019} are proposed.
    LMB is the computationally efficient alternative to $\delta$-GLMB that is identical to LMB mixture under the LMB birth~\cite[Sec.~IV]{Garcia-Fernandez2018}.
    
    Poisson multi-Bernoulli mixture~(PMBM) SLAM~\cite{Yu_Cluster_Sensors,Yu_PMBM_GLO2020,Yu_EK-PMBM_JSAC2021} was introduced without proper mathematical justification for the sensor state density computation.
    The $\delta$-GLMB~\cite{vo2013LMB} and PMBM~\cite{Garcia-Fernandez2018} densities are conjugate priors for each Bayesian recursion and can handle the entire SLAM problem in a theoretically optimal manner.
    Compared to $\delta$-GLMB, PMBM efficiently represents the set of landmarks as two disjoint subsets: undetected landmarks, modeled by Poisson; and detected landmarks, modeled by multi-Bernoulli mixture~(MBM), leading to computational savings in parameterization and hypotheses cardinality~\cite[Sec.~IV]{Garcia-Fernandez2018}.
    Therefore, the PMBM filter is attractive as a starting point for reduced-complexity variations, and is the focus of this work.}
    A reduced-complexity Poisson multi-Bernoulli~(PMB) joint sensor and target tracking filter was proposed in~\cite{Froehle19}, based on the reduction from PMBM to PMB using belief propagation~(BP), but required several ad-hoc steps.  
    BP has also been used to derive low-complexity PMB or MB SLAM  directly (known as BP-SLAM~\cite{Rico_BPSLAM_JSTSP2019,Erik_BPSLAM_TWC2019,Erik_AOABPSLAM_ICC2019}), by modeling the landmarks with random vectors rather than random sets and computing marginal posteriors of landmarks and sensor state. 
    In this paper, we provide a rigorous derivation of the PMBM SLAM filter for a scenario involving mixed continuous and discrete states, to capture different landmark models. From the 
    PMBM-SLAM filter, we develop a reduced complexity PMB-SLAM filter and a low-complexity marginalized PMB-SLAM filter. 
    The main contributions of this paper are as follows:
\begin{itemize}
    \item We provide the first complete derivation of the PMBM-SLAM filter from the PMBM-MTT filter, by adopting the expected RFS-likelihood calculation from~\cite{Garcia-Fernandez2018}.
    It provides a fully Bayesian solution to the complete problem, including the birth of landmarks and a compact representation of all hypotheses, where undetected landmarks are elegantly modeled by means of a Poisson point process~(PPP).
    \item We derive the PMB-SLAM filter from the proposed PMBM-SLAM filter,
    motivated by the introduction of a novel auxiliary variable of data association.
    We also reveal a novel connection with the Bethe free energy~\cite{Yedidia2005BetheFree} to the computation of the normalization constant for determining the sensor particle weights.
    \item We develop a novel marginalized PMB-SLAM filter from the PMB-SLAM filter, by marginalizing out the sensor state and global hypotheses,
    and by deriving the marginalized sensor posterior density.
    By not tracking the correlations between the sensor state and landmarks, a significant reduction in the computational burden is achieved while sustaining the SLAM accuracy compared to both PMBM and PMB SLAM filters. {We also reveal close connections of the BP-SLAM filter to the proposed marginalized PMB-SLAM filter.}
    \item
    The developed marginalized PMB-SLAM filter provides a new framework for developing new versions of EKF-SLAM~\cite{Gamini_EKFSLAM,Randall_EKFSLAM}, FastSLAM~\cite{FastSLAM,FastSLAM2,Gentner_SLAM_TWC2016}, and GraphSLAM~\cite{GraphSLAM_2006,GraphSLAM}, with an inherent ability to acknowledge the association uncertainties, i.e., it simultaneously solves the front-end and back-end problems in a theoretically optimal manner.
    \item We validate the three proposed SLAM filters in 5G mmWave vehicular networks 
    and provide a performance comparison with the PHD-SLAM filter~\cite{Hyowon_TWC2020}.
\end{itemize}
    
    The rest of this paper is organized as follows.
    Section~\ref{sec:Models} presents the system model in vehicular networks with 5G mmWave communication links.
    Section~\ref{sec:PMBM} introduces the backgrounds of RFS-based SLAM.
    In Section~\ref{sec:SLAMFilters}, the PMBM and PMB SLAM filters are derived, and the marginalized PMB-SLAM filter is proposed.
    Section~\ref{sec:Implementation} presents the CKF implementation of marginalized PMB-SLAM.
    The numerical results and discussions are reported in Section~\ref{sec:Results}, and conclusions are drawn in Section~\ref{sec:Conclusions}. 
    
 \subsubsection*{Notation}
     Scalars are indicated by the italic font, e.g., $x$.
     Vectors and matrices are respectively displayed in the bold lowercase and uppercase letters, e.g., $\mathbf{x}$ and $\mathbf{X}$, and their transpose are indicated by superscript $\top$, e.g., $\mathbf{x}^\top$ and $\mathbf{X}^\top$.
     Sets are displayed in the calligraphic font, e.g., $\mathcal{X}$, and the cardinality of set $\mathcal{X}$ is denoted by $\lvert \mathcal{X} \rvert$.
     The probability density function (pdf) and probability mass function (pmf) are respectively denoted $f(\cdot)$ and $\mathsf{p}(\cdot)$. 
     The number of unknown variables of $\mathbf{x}$ is denoted by $\mathsf{d}{(\mathbf{x})}$.
     The disjoint union of sets is denoted by $\uplus$.
     The symbol $\appropto$ stands for approximately proportional to {application-specific notations are described in Table \ref{Tab:Notations}.}

\begin{table}
    \centering
    \label{Tab:Notations}
    \caption{{Common Notations}}
    \vspace{-2mm}
    \resizebox{0.97\columnwidth}{!} {
\begin{tabular}{|p{14mm}|p{33mm}||p{14mm}|p{33mm}|}
    \hline
    \textbf{Notation} & \textbf{Description} & \textbf{Notation} & \textbf{Description}
    \tabularnewline
    \hline 
    $k$ & time index & $n$ & particle index
    \tabularnewline
    \hline 
    $i$ & landmark index & $j$ & measurement index
    \tabularnewline
    \hline 
    $\mathbf{s}_k$ & vehicle state & 
    $f_k$ & vehicle density
    \tabularnewline
    \hline 
    $\mathbf{x}$ & landmark location &
    $m$ & landmark type
    \tabularnewline
    \hline 
    $a_k^{i,n}$ & local hypothesis~(LH) & $\mathbf{a}_k \in \mathcal{A}_{k}^n$ & global hypothesis~(GH) 
    \tabularnewline
    \hline 
    $r^{i,a_k^i,n}_k$ & existence probability & $f^{i,a_k^i,n}_k$ & landmark density
    \tabularnewline
    \hline 
    $\beta_k^{i,{a}^i_k,n}$ & LH weight & 
    $\beta_k^{\mathbf{a}_k,n}$ & GH weight
    \tabularnewline
    \hline 
\end{tabular}}\vspace{-5mm}
\end{table}

    
\section{System Model in 5G Vehicular Networks}\label{sec:Models}
    Without loss of generality, we describe models for the vehicle dynamics and observations, and introduce the propagation environment, consisting of different types of landmarks.
    We consider that a vehicle is moving around a base station~(BS), and a vehicle state at time step $k$ is denoted by $\mathbf{s}_{k}=[\mathbf{v}_k^\top,{\alpha}_{k},\zeta_{k},\rho_{k},b_{k}]^\top$, where 
    $\mathbf{v}_k=[x_{\text{v},k}, y_{\text{v},k}, z_{\text{v},k}]^\top$, ${\alpha}_{k}$, $\zeta_{k}$,  $\rho_{k}$, and $b_k$ are respectively the 3-dimensional (3D) location, heading, translation speed, angular speed, and clock bias.
    With the known prior density $f(\mathbf{s}_0)$ and known transition density $f(\mathbf{s}_k|\mathbf{s}_{k-1})$, we model the vehicle dynamics as
\begin{align} \label{eq:mobility}
    \mathbf{s}_k = \mathsf{v}(\mathbf{s}_{k-1}) + \mathbf{q}_k, 
\end{align}
    where $\mathsf{v}(\cdot)$ is a transition function~\cite[Chapter 5]{thrun2005probabilistic},~\cite{RongLiJ:2003}, and $\mathbf{q}_k\sim \mathcal{N}(\bf{0},\mathbf{Q})$ denotes the process noise with the known covariance matrix 
    $\mathbf{Q}$.
    
    {The BS periodically transmits the mmWave signals.
    We consider the following propagation environments~\cite{Hyowon_TWC2020} in the viewpoint of the vehicle receiver: i) direct path from BS; ii) reflected path from the reflection surfaces, characterized as virtual anchors~(VAs) by mirroring the BS to surface~\cite{Witrisal_SPM2016,RicoTWC2018,HenkGlobecom2018,Rappaport_6G100GHz_Access2019}; and iii) scattered path from small objects, modeled as scattering points~(SPs), exist as shown in Fig.~\ref{Fig:scenario}.}
    We regard the BS, VAs, and SPs as landmarks which are static, and we assume that landmarks never appear or disappear in the propagation environment.
    We denote the landmark type by $m \in \mathcal{M}=\{\text{BS},\text{VA},\text{SP}\}$ and the landmark location by $\mathbf{x}=[x,y,z]^\top$.
    The set of all landmark locations and types is modeled as an RFS $\mathcal{X}$ with the set density $f(\mathcal{X})$.
    

    
    The receiver at the vehicle can detect the signals~\cite{Heath_MIMOChannel_JSTSP2016}.
    Each signal comes from the landmark $(\mathbf{x},m) \in \mathcal{X}$ with an adaptive detection probability~\cite{wymeersch2020adaptive}, denoted by $\mathsf{p}_{\text{D},k}(\mathbf{s}_k,\mathbf{x},m) \in [0,1]$, which depends on the FoV, and both vehicle and landmark states.
    We assume a channel estimation routine is performed at the vehicle receiver, which provides a bunch of measurements.
    Then, using the channel estimation routine the measurement for the landmark $(\mathbf{x},m)\in\mathcal{X}$ is provided as
\begin{align}
    \mathbf{z}_{k}^j = \mathsf{h}(\mathbf{s}_k, \mathbf{x}, m) + \mathbf{r}_{k}^j, \label{eq:measurementModel}
\end{align}
    where $\mathsf{h}(\mathbf{s}_k, \mathbf{x}, m) = [\tau_{k}^j, (\bm{\theta}_{k}^j)^\top, (\bm{\phi}_{k}^j)^\top]^\top$, and $\mathbf{r}_{k}^j\sim \mathcal{N}(\bm{0},{\mathbf{R}_k^j})$ denotes the measurement noise.
    Here, $\tau_{k}^j$, $\bm{\theta}_{k}^j$, and $\bm{\phi}_{k}^j$ respectively denote a time-of-arrival (TOA), direction-of-arrival (DOA) in azimuth and elevation, direction-of-departure (DOD) in azimuth and elevation, which follow the geometric relations, see, e.g.,~\cite[Appendix B]{Hyowon_TWC2020}, and ${\mathbf{R}_k^j}$ denotes a known covariance matrix.
    We  define clutter 
    consisting of either false alarms or transient targets (e.g., people or passing car) or multi-bounce signals.
    We model the set of measurements as an RFS $\mathcal{Z}_k=\{\mathbf{z}_{k}^1,...,\mathbf{z}_{k}^{J_k}\}$, where $J_k$ is the number of observable measurements (including clutter) at the vehicle receiver.

    The goal is to determine the joint posterior of the vehicle state and landmarks given the prior $f(\mathbf{s}_0)$ and the BS location, i.e., $f(\mathbf{s}_{0:k},\mathcal{X}|\mathcal{Z}_{1:k})$, and to develop a low complexity SLAM filter which computes $f(\mathbf{s}_{k},\mathcal{X}|\mathcal{Z}_{1:k})$.

    

\section{{Background for RFS-based SLAM Filter}}\label{sec:PMBM}
    We provide background for RFS-based SLAM filters.
    First, we briefly present the basic computation 
    and densities of RFS.
    Second, we introduce two joint densities for SLAM and the Bayesian recursion for SLAM. 
    Third, we also introduce marginalization of a PMB density. 

\subsection{Basics on RFS Densities}
\label{sec:RFSdensity}
\subsubsection{General RFS}\label{sec:Background}
    Let us denote a finite set by $\mathcal{X} = \{\mathbf{x}^{1},...,\mathbf{x}^{I}\}$, where each vector $\mathbf{x}^i \in \mathbb{R}^{n_\mathbf{x}}$ is random, and the cardinality $I=\lvert \mathcal{X} \rvert$ is also random, where $\lvert \cdot \rvert$ is the set cardinality.
    By FISST~\cite{mahler_book_2007} with the joint pdf $f(\mathbf{x}^{1},...,\mathbf{x}^{I})$, the set density $f(\mathcal{X})$ is given by~\cite{williams2015marginal}
    $f(\{\mathbf{x}^{1},...,\mathbf{x}^{I}\}) = \mathsf{p}(I)\sum_\pi f(\mathbf{x}^{\pi(1)},...,\mathbf{x}^{\pi(I)})$,
    where $\mathsf{p}(I)=\text{Pr}(\lvert \mathcal{X} \rvert=I)$, and $\pi(\cdot)$ denotes a permutation function, which indicates that the joint pdf $f(\cdot)$ is invariant to all permutations. The set integral is defined in~\cite{mahler_AES_2003_PHD}.

\subsubsection{Poisson RFS}\label{sec:BasicsPPP}
    The RFS $\mathcal{X}$ follows a Poisson process: the cardinality $I$ is Poisson distributed with mean $\mu$ (i.e., $\mathsf{p}(I)=\mu^I \exp(-\mu)/I!$); and given the cardinality $I$, each vector $\mathbf{x}^i\in \mathcal{X}$ is independent and identically distributed with the density $f(\mathbf{x})$ (i.e., $f(\mathbf{x}^{\pi(1)},...,\mathbf{x}^{\pi(I)})=\prod_{\mathbf{x}\in \mathcal{X}} f(\mathbf{x})$).
    Then, the intensity function of the Poisson RFS is given by $\lambda(\mathbf{x})=\mu f(\mathbf{x})$, and the density of the Poisson RFS $\mathcal{X}$ is given by~\cite[pp.~373]{williams2015marginal}
\begin{align} \label{eq:PoissonDensity}
    f(\mathcal{X}) = e^{-\int \lambda(\mathbf{x})\mathrm{d}\mathbf{x}}\prod_{\mathbf{x}\in\mathcal{X}}\lambda(\mathbf{x}).
\end{align}

\subsubsection{Bernoulli RFS and its Generalizations}\label{sec:BasicsMBP}
    The density of the RFS $\mathcal{X}$ following a Bernoulli process is given by
\begin{align} \label{eq:Bernoulli}
    f(\mathcal{X}) =
    \begin{cases}
    1-r,& \mathcal{X}=\emptyset,\\
    rf(\mathbf{x}), & \mathcal{X}=\{\mathbf{x}\},\\
    0, & \lvert \mathcal{X} \rvert > 1,
    \end{cases}
\end{align}
    where $r \in [0,1]$ denotes the landmark's existence probability.
    When $\mathcal{X}$ is the union of $I$ independent Bernoulli RFSs $\mathcal{X}^i,~i=1,\ldots,I$ with the pdf $f^i(\mathcal{X}^i)$, defined in~\eqref{eq:Bernoulli}, the RFS $\mathcal{X}$ follows a multi-Bernoulli (MB) process.
    Then, using convolution theorem for independent RFSs~\cite[pp.~372, 386]{mahler_book_2007}, the density of the RFS $\mathcal{X}$
    is represented as
\begin{align} \label{eq:MBdensity}
    f(\mathcal{X}) = \sum_{\uplus_{i=1}^{I}\mathcal{X}^{i}=\mathcal{X}}\prod_{i=1}^{I} f^i(\mathcal{X}^{i}),
\end{align}
    where $\uplus_{i=1}^{I}\mathcal{X}^{i}$ indicates $\mathcal{X}^{1}\uplus \cdots \uplus \mathcal{X}^{I}$. Here, $\uplus$ stands for the disjoint set union.
    Finally, a multi-Bernoulli mixture (MBM) RFS $\mathcal{X}$ is expressed as a linear combination of $\lvert \mathcal{A} \rvert$ MB densities~\cite{williams2015marginal,Garcia-Fernandez2018}, where $\mathcal{A}$ denotes a set of global hypotheses:
\begin{align} \label{eq:MBMdensity}
    f(\mathcal{X}) 
    & = \sum_{\mathbf{a}\in\mathcal{A}}\beta^\mathbf{a}\,\,\,\sum_{\uplus_{i=1}^{I}\mathcal{X}^{i}=\mathcal{X}}\,\,\,\prod_{i=1}^{I} f^{i,a^i}(\mathcal{X}^{i}),\\
    &\propto \sum_{\mathbf{a} \in \mathcal{A}} \,\,\,\sum_{\uplus_{i=1}^{I}\mathcal{X}^{i}=\mathcal{X}}\,\,\,\prod_{i=1}^{I} \beta^{i,a^i}f^{i,a^i}(\mathcal{X}^{i}),
\end{align}
    where $\mathbf{a}=[a^1,...,a^I]$ denotes a single global hypothesis, where $a^i$ is the local hypothesis indicating the single-trajectory  hypothesis for $i$-th Bernoulli, i.e., the measurements associated at different times to this Bernoulli; $\beta^\mathbf{a}\propto \prod_{i=1}^I\beta^{i,a^i}$ is the weight of global hypothesis $\mathbf{a}$, such that $\sum_{\mathbf{a}\in \mathcal{A}}\beta^\mathbf{a}=1$, i.e., $\mathsf{p}(\mathbf{a})=\beta^\mathbf{a}$; and
    $\beta^{i,a^i}$ and $f^{i,a^i}(\mathcal{X}^i)$ are the weight and the Bernoulli density of potentially detected landmark $i$ under hypothesis $a^i$.

\subsubsection{PMBM and PMB}\label{sec:PMB_PMBM}
    If $\mathcal{X}^{\text{U}}$ and $\mathcal{X}^{\text{D}}$ are independent RFSs such that $\mathcal{X} = \mathcal{X}^{\text{U}} \uplus \mathcal{X}^{\text{D}}$, then 
    the set density $f(\mathcal{X})$ is 
\begin{align}\label{eq:PMBMDensity}
    f(\mathcal{X}) = \sum_{\mathcal{X}^\mathsf{U} \uplus \mathcal{X}^\text{D} =\mathcal{X}} f^{\text{U}}(\mathcal{X}^\text{U})f^{\text{D}}(\mathcal{X}^\text{D}).
\end{align}    
    When $\mathcal{X}^\text{U}$ and $\mathcal{X}^\text{D}$ respectively follow the Poisson process in~\eqref{eq:PoissonDensity} and MBM process in~\eqref{eq:MBMdensity}, then $f(\mathcal{X})$ is a PMBM density.
    Substituting~\eqref{eq:PoissonDensity}--\eqref{eq:MBMdensity} into~\eqref{eq:PMBMDensity}~\cite{Garcia-Fernandez2018},
\begin{align}\label{eq:PMBdensity}
    f(\mathcal{X}) 
    \propto \sum_{\uplus_{i=1}^{I}\mathcal{X}^{i}\uplus \mathcal{X}^\text{U} =\mathcal{X}} \,\,\prod_{\mathbf{x}\in \mathcal{X}^\text{U}}\lambda(\mathbf{x})
    \sum_{\mathbf{a}\in\mathcal{A}} \prod_{i=1}^{I}\beta^{i,a^i}f^{i,a^i}(\mathcal{X}^i).
\end{align}
    Note that when $\mathcal{X}^{\text{D}}$ follows an multi-Bernoulli~(MB) process (i.e., $\lvert \mathcal{A}\rvert=1$), $f(\mathcal{X})$ is a PMB density $ f(\mathcal{X}) \propto \sum_{\uplus_{i=1}^{I}\mathcal{X}^{i}\uplus \mathcal{X}^\text{U} =\mathcal{X}} \,\,\prod_{\mathbf{x}\in \mathcal{X}^\text{U}}\lambda(\mathbf{x}) \prod_{i=1}^{I}f^{i}(\mathcal{X}^i).$

\subsection{Bayesian Recursion of RFS-joint SLAM Density}\label{sec:BayesianRecursion}

{We now describe the SLAM recursion. }
\subsubsection{Joint Vehicle Trajectory and Landmark Density}\label{sec:TrajecRecur}
    At time step $k$, a joint posterior density for a vehicle trajectory $\mathbf{s}_{0:k}$ and a set of landmarks $\mathcal{X}$ is denoted by $f(\mathbf{s}_{0:k},\mathcal{X}|\mathcal{Z}_{1:k})$, which can be factorized as 
\begin{align}\label{eq:TTSLAMden}
    f(\mathbf{s}_{0:k},\mathcal{X}|\mathcal{Z}_{1:k}) = f(\mathbf{s}_{0:k}|\mathcal{Z}_{1:k})f(\mathcal{X}|\mathbf{s}_{0:k},\mathcal{Z}_{1:k}),
\end{align}
    where $f(\mathbf{s}_{0:k}|\mathcal{Z}_{1:k})$ and $f(\mathcal{X}|\mathbf{s}_{0:k},\mathcal{Z}_{1:k})$ are respectively posterior densities for the vehicle trajectory and the set of landmarks conditioned on the vehicle trajectory.
    Each density goes through the following prediction and update steps.
    %
    Assume that the motion of $\mathbf{s}_k$ is independent of $\mathcal{X}$, and that the targets are static. The vehicle trajectory is predicted as
\begin{align} \label{eq:BasicVehPred}
    f(\mathbf{s}_{0:k}|\mathcal{Z}_{1:k-1}) =  f(\mathbf{s}_k|\mathbf{s}_{k-1})f(\mathbf{s}_{0:k-1}|\mathcal{Z}_{1:k-1}),
\end{align}
    where $f(\mathbf{s}_k|\mathbf{s}_{k-1})$ is the known transition density of the dynamic model~\eqref{eq:mobility}.
    We remind the reader that the static landmarks are assumed to never appear or disappear. Thus, the landmarks have no prediction~\cite{Maryam_PMB_TSP2017} and $f(\mathcal{X}|\mathbf{s}_{0:k},\mathcal{Z}_{1:k-1}) = f(\mathcal{X}|\mathbf{s}_{0:k-1},\mathcal{Z}_{1:k-1})$.
    %
    The set of landmarks conditioned on the vehicle trajectory is updated as
\begin{align} \label{eq:BasicObjUp}
    f(\mathcal{X}|\mathbf{s}_{0:k},\mathcal{Z}_{1:k}) = \dfrac{f(\mathcal{X}|\mathbf{s}_{0:k},\mathcal{Z}_{1:k-1})g(\mathcal{Z}_k|\mathbf{s}_{0:k},\mathcal{X},\mathcal{Z}_{1:k-1})}
    {g(\mathcal{Z}_k|\mathbf{s}_{0:k},\mathcal{Z}_{1:k-1})},
\end{align}
    where $g(\mathcal{Z}_k|\mathbf{s}_{0:k},\mathcal{X},\mathcal{Z}_{1:k-1})$ is the RFS-likelihood of the measurement set $\mathcal{Z}_k$ for $\mathbf{s}_{0:k}$ and $\mathcal{X}$, and $g(\mathcal{Z}_k|\mathbf{s}_{0:k},\mathcal{Z}_{1:k-1})$ is the normalizing factor.
    The vehicle trajectory is updated as
\begin{align} \label{eq:BasicVehUp}
    f(\mathbf{s}_{0:k}|\mathcal{Z}_{1:k}) = \dfrac{f(\mathbf{s}_{0:k}|\mathcal{Z}_{1:k-1})g(\mathcal{Z}_k|\mathbf{s}_{0:k},\mathcal{Z}_{1:k-1})}
    {g(\mathcal{Z}_k|\mathcal{Z}_{1:k-1})},
\end{align}
    where 
    $g(\mathcal{Z}_k|\mathcal{Z}_{1:k-1})$ is the normalizing factor.

\subsubsection{Marginal Vehicle State and Landmark Density}\label{sec:StateRecur}
    In case we are not interested in vehicle trajectories, we can reduce complexity by recursively determining marginal densities $f(\mathbf{s}_k|\mathcal{Z}_{1:k})$ and $f(\mathcal{X}|\mathcal{Z}_{1:k})$.
    Assume $f(\mathbf{s}_{k-1}|\mathcal{Z}_{1:k-1})$ and $f(\mathcal{X}|\mathcal{Z}_{1:k-1})$ are given, then $f(\mathbf{s}_k|\mathcal{Z}_{1:k-1}) = \int f(\mathbf{s}_k|\mathbf{s}_{k-1})f(\mathbf{s}_{k-1}|\mathcal{Z}_{1:k-1})\mathrm{d}\mathbf{s}_{k-1}$ by the Chapman-Kolmogorov equation.
    We recall that the landmarks have no prediction. The update step then becomes~\cite{Froehle19}
\begin{align}
    &f(\mathbf{s}_k|\mathcal{Z}_{1:k}) \propto f(\mathbf{s}_k|\mathcal{Z}_{1:k-1})\int  f(\mathcal{X}|\mathcal{Z}_{1:k-1})  g(\mathcal{Z}_k|\mathbf{s}_k,\mathcal{X}) \delta \mathcal{X}, \label{eq:UpMarVeh}
\end{align}
and
\begin{align}
    &f(\mathcal{X}|\mathcal{Z}_{1:k}) \propto f(\mathcal{X}|\mathcal{Z}_{1:k-1}) \int f(\mathbf{s}_k|\mathcal{Z}_{1:k-1}) g(\mathcal{Z}_k|\mathbf{s}_k,\mathcal{X}) \mathrm{d}\mathbf{s}_{k}\label{eq:UpMarObj}.
\end{align}

\subsection{{Marginal PMB Density}}\label{sec:marginalPMBMethod}
    
    In our proposed filters, we will work with PMB densities that are conditioned on nuisance variables: the global hypotheses (see \eqref{eq:PMBdensity}) or the vehicle state (see \eqref{eq:BasicObjUp}). When we marginalize out these nuisance variables and approximate again with a PMB density (leading to a marginal PMB density), less complex filters result, since the marginal PMB density of the set of landmarks can be represented as the Poisson and MB components.
    In the following, the marginal PMB density is approximated, derived by the Kullback-Leibler divergence~(KLD) method.
    
    Given a density $f(\mathcal{X}|\bm{\eta})$, conditioned on a nuisance variable $\bm{\eta}$ that can contain continuous and discrete variables with the probability distribution $p(\bm{\eta})$, of the form of a PMB, $f(\mathcal{X}|\bm{\eta})= \sum_{\uplus_{i=1}^{I}\mathcal{X}^{i}\uplus \mathcal{X}^\text{U} =\mathcal{X}} \,\,f^{\text{U}}(\mathcal{X}^\text{U}|\bm{\eta}) \prod_{i=1}^{I}f^{i}(\mathcal{X}^i|\bm{\eta})$. In case we want to approximate $f(\mathcal{X})=\mathbb{E}_{\bm{\eta}}[f(\mathcal{X}|\bm{\eta})]$
    with a PMB, we use the following approach~\cite{Garcia_Trajectory_2020}: we extend the state space with an auxiliary variable by $u \in \mathcal{U}=\{0,1,...,I\}$, where $u=0$ implies that the landmark has not yet been detected, while $u=i>0$ indicates that the landmark corresponds to the $i$-th Bernoulli component. A set of landmark states with auxiliary variables is denoted by $\tilde{\mathcal{X}}$ with elements $(u,\mathbf{x}) \in \mathcal{U}\times \mathbb{R}^{n_\mathbf{x}}$.
    Then, we define
\begin{align}
    \tilde{f}(\tilde{\mathcal{X}})&=\mathbb{E}_{\bm{\eta}}[\tilde{f}^{\text{U}}(\tilde{\mathcal{X}}^\text{U}|\bm{\eta})\prod_{i=1}^{I}\tilde{f}^{i}(\tilde{\mathcal{X}}^i|\bm{\eta})],
    \label{eq:PMBwithNuisance}
\end{align}
where
$\tilde{\mathcal{X}}^\text{U}=\{(u,\mathbf{x}) \in \tilde{\mathcal{X}}: u=0 \}$ and $\tilde{\mathcal{X}}^i=\{(u,\mathbf{x}) \in \tilde{\mathcal{X}}: u=i \}$, and similarly to~\cite[eq.~(8)]{Garcia_Trajectory_2020} we express
\begin{align}
    \tilde{f}^{\text{U}}(\tilde{\mathcal{X}}|\bm{\eta})&
    = \exp(-\int \lambda(\mathbf{x}|\bm{\eta})\textrm{d}\mathbf{x}) \prod_{(u,\mathbf{x})\in \tilde{\mathcal{X}}}\delta_{0,u}\lambda(\mathbf{x}|\bm{\eta}),\\
    \tilde{f}^{i}(\tilde{\mathcal{X}}|\bm{\eta})&=
    \begin{cases}
        1-r^i(\bm{\eta}), & \tilde{\mathcal{X}}=\emptyset,\\
        r^i(\bm{\eta}) f^{i}(\mathbf{x}|\bm{\eta})\delta_{u,i}, & \tilde{\mathcal{X}}=(u,\mathbf{x}),\\
        0, & \text{otherwise},
    \end{cases}
\end{align}
    where $\delta_{u,i}$ denotes a Kronecker delta, defined as $\delta_{u,i}=1$ if $u=i$ and $\delta_{u,i}=0$, otherwise, and the existence probability is depending on the nuisance variable, denoted by $r^i(\bm{\eta})$.
    It is readily verified that marginalizing out the auxiliary variables in $\tilde{f}(\tilde{\mathcal{X}})$ yields $\mathbb{E}_{\bm{\eta}} [f(\mathcal{X}|\bm{\eta})]$.
    The goal is to obtain a PMB approximation $ \tilde{q}(\tilde{\mathcal{X}})=\tilde{q}^{\text{U}}(\tilde{\mathcal{X}}^\text{U})\prod_{i=1}^{I}\tilde{q}^{i}(\tilde{\mathcal{X}}^i)$, 
    and we express
\begin{align}
    \tilde{q}^{\text{U}}(\tilde{\mathcal{X}})&
    = \exp\big(-\int \lambda^q(\mathbf{x})\textrm{d}\mathbf{x}\big) \prod_{(u,\mathbf{x})\in \tilde{\mathcal{X}}}\delta_{0,u}\lambda^q(\mathbf{x})
    \end{align}
    \begin{align}
    \tilde{q}^{i}(\tilde{\mathcal{X}})&=
    \begin{cases}
        1-r^{i}, & \tilde{\mathcal{X}}=\emptyset,\\
        r^{i} f^{i}(\mathbf{x})\delta_{u,i}, & \tilde{\mathcal{X}}=(u,\mathbf{x}),\\
        0, & \text{otherwise}.
    \end{cases}
\end{align}

\begin{table*}
\centering
\newcolumntype{P}[1]{>{\centering\arraybackslash}m{#1}}
\caption{Overview of the three proposed SLAM filters ($N$ particle samples; $\mathcal{A}_k^n$ updated global hypotheses; $\lvert \mathcal{Z}_k \rvert = J_k$ newly detected landmarks or clutter; $I_{k-1}$ previously detected landmarks; $B_{\max}$ maximum allowable global hypotheses; $L_{\max}$ iterations, required for BP convergence).} \label{tab:Overview}
\resizebox{0.9\textwidth}{!} {
\begin{tabular}{|l|c|c|c|}
    \hline 
    &
    \begin{tabular}{@{}c@{}}
    \textbf{PMBM-SLAM} 
    \end{tabular}
    &
    \begin{tabular}{@{}c@{}}
    \textbf{PMB-SLAM} 
    \end{tabular}
    &
    \begin{tabular}{@{}c@{}}
    \textbf{Marg.~PMB-SLAM} 
    \end{tabular}
    \tabularnewline
    \hline 
    \hline 
    {{Bayesian recursion}} & {$f(\mathbf{s}_{0:k},\mathcal{X}|\mathcal{Z}_{1:k})$ of~\eqref{eq:TTSLAMden}} & {$f(\mathbf{s}_{0:k},\mathcal{X}|\mathcal{Z}_{1:k})$ of~\eqref{eq:TTSLAMden}} & {$f(\mathbf{s}_{k}|\mathcal{Z}_{1:k})$ of~\eqref{eq:UpMarVeh}, $f(\mathcal{X}|\mathcal{Z}_{1:k})$ of~\eqref{eq:UpMarObj}}
    \tabularnewline
    \hline 
    Sensor Repr.
    &
    \begin{tabular}{@{}c@{}}
    $f(\mathbf{s}_{0:k}|\mathcal{Z}_{1:k})\approx \sum_{n=1}^N \delta(\mathbf{s}_{0:k}-\mathbf{s}_{0:k}^n) w_k^n$ \\ (by particle sample)
    \end{tabular}
    &
    \begin{tabular}{@{}c@{}}
    $f(\mathbf{s}_{0:k}|\mathcal{Z}_{1:k})\approx \sum_{n=1}^N \delta(\mathbf{s}_{0:k}-\mathbf{s}_{0:k}^n) w_k^n$
    \\
    (by particle sample)
    \end{tabular}
    &
    \begin{tabular}{@{}c@{}}
    $f(\mathbf{s}_{k}|\mathcal{Z}_{1:k})\approx \mathcal{N}(\mathbf{s}_k;\mathbf{s}_{\mathsf{u},k},\mathbf{U}_{\mathsf{u},k})$
    \\
    (by CKF)
    \end{tabular}
    \tabularnewline
    \hline 
    Map Repr.
    
    & $f(\mathcal{X}|\mathbf{s}_{0:k}^n,\mathcal{Z}_{1:k})$ by PMBM, $\forall n$  & $f(\mathcal{X}|\mathbf{s}_{0:k}^n,\mathcal{Z}_{1:k})$ by PMB, $\forall n$ & $f(\mathcal{X}|\mathcal{Z}_{1:k})$ by PMB
    \tabularnewline
    \hline
    Data Assoc. &
    \begin{tabular}{@{}c@{}}
    update $\mathcal{A}_k^n$ by $B_\text{max}$-best global hypotheses
    \\
    using Murty's alg.~\cite{Garcia-Fernandez2018,murthy1968algorithm}
    \end{tabular}
    &
    \begin{tabular}{@{}c@{}}
    compute marginal association \\ probabilities~(Sec.~\ref{sec:AppPMBMtoaPMB})
    \\
    by BP during $L_\text{max}$ iterations~\cite{williams2015marginal}
    \end{tabular}
    &
    \begin{tabular}{@{}c@{}}
    compute marginal association
    \\ 
    probabilities~(Sec.~\ref{sec:AppPMBMtoaPMB})
    \\
    by BP during $L_\text{max}$ iterations~\cite{williams2015marginal}
    \end{tabular}
    \tabularnewline
    \hline 
    {\# Global Hypo.}
    &
    $\lvert \mathcal{A}_k^n \rvert \geq 1,~\forall~n$
    & 
    $\lvert \mathcal{A}_k^n \rvert =1,~\forall~n$
    & 
    $\lvert \mathcal{A}_k \rvert =1$
    \tabularnewline
    \hline 
    Vehicle weight
    &
    \begin{tabular}{@{}c@{}}
    compute~\eqref{eq:ELCompute} by the updated MBM \\ components~(Sec.~\ref{sec:PMBMpred})
    \end{tabular}
    & compute~\eqref{eq:ELCompute} by Bethe free energy~\eqref{eq:BTFapp} & --\tabularnewline
    \hline 
     \begin{tabular}{@{}c@{}}
    Complexity per
    \\
    time step $k$
    \end{tabular}
    &
      $\mathcal{O}\left(\sum_n \vert\mathcal{A}^n_{k-1}\vert \left(I_{k-1} + J_k\right)^3B_{\max} \right)$
      &
      $\mathcal{O}\left(N I_{k-1} J_k L_{\max}\right)$
      &
      $\mathcal{O}\left(I_{k-1} J_k L_{\max}\right)$
      \tabularnewline
    \hline 
\end{tabular}}
\end{table*}

\begin{lemma} \label{lem:MinKLD}
    \normalfont
    Given the density $\tilde{f}(\tilde{\mathcal{X}})$ of the form~\eqref{eq:PMBwithNuisance}, the PMB approximation $\tilde{q}(\tilde{\mathcal{X}})$
    that minimizes the KLD $D(\tilde{f} \Vert \tilde{q})$ is of the form $\tilde{q}^{\text{U}}(\tilde{\mathcal{X}}^{\text{U}})  = \mathbb{E}_{\bm{\eta}} [\tilde{f}^{\text{U}}(\tilde{\mathcal{X}}^\text{U}|\bm{\eta})]$ and $ \tilde{q}^{i}(\tilde{\mathcal{X}}^i)  = \mathbb{E}_{\bm{\eta}} [\tilde{f}^{i}(\tilde{\mathcal{X}}^i|\bm{\eta})]$.
\end{lemma}
\begin{proof}
    The proof is similar to~\cite[Appendix B]{Garcia_Trajectory_2020}.
\end{proof}
    Given a PMBM posterior (where we obtain a mixture due to the nuisance parameter $\bm{\eta}$), we can thus introduce auxiliary variables, use Lemma~\ref{lem:MinKLD} to approximate that PMBM distribution as a PMB distribution with auxiliary variable, and finally marginalize out the nuisance parameters to obtain a PMB without auxiliary variables.
    For example, when $\bm{\eta}$ is continuous variable then the parameters of the final PMB distribution can be expressed in terms of the original PMBM as follows.
    Using the sense of KLD minimization in Lemma~\ref{lem:MinKLD}, the PMB components are computed as follows: the Poisson intensity is given by $\lambda^q(\mathbf{x})=\int p(\bm{\eta})\lambda(\mathbf{x}|\bm{\eta}) \mathrm{d}\bm{\eta}$, while 
    the Bernoulli components have existence probability  $r^{i} = \int p(\bm{\eta}) r^i(\bm{\eta}) \mathrm{d}\bm{\eta}$ and  density $   f^i(\mathbf{x}) = {\int p(\bm{\eta}) r^i(\bm{\eta}) f^i(\mathbf{x}|\bm{\eta}) \mathrm{d}\bm{\eta}}/({\int p(\bm{\eta}) r^i(\bm{\eta}) \mathrm{d}\bm{\eta}})$. 
    In the approximation of the PMBM as a PMB, we merge all possible distributions for the PPP into a new PPP, and we also merge all possible distributions for Bernoulli component $i$ (landmark $i$) into Bernoulli component $i$~\cite[Proposition~2]{Garcia_Trajectory_2020}.

\section{Proposed RFS-based Filters}\label{sec:SLAMFilters}
    In this section, we provide an intuition of the proposed SLAM filters. This is followed by the proposed PMBM, PMB, and marginalized PMB SLAM filters. {The main results relate to the vehicle state update, and are expressed as Propositions 1--3.}

\subsection{Short Description of the Proposed SLAM Filters}
\label{sec:PMBintution}
    {The main development objective is to provide a low-complexity implementation of a PMBM-based SLAM filter.}
    In Table~\ref{tab:Overview}, we provide a high-level comparison of the three proposed filters.
    The starting point is the PMBM-SLAM filter (see Section \ref{sec:PMBM-SLAM-Filter}), implemented by an RBPF~\cite{Mullane2011}. The entire filter is described for completeness, though our novelty pertains to the computation of the particle weights. 
    From the PMBM-SLAM filter, two PMB SLAM filters are developed by marginalizing out the nuisance parameter (e.g., global hypothesis $\mathbf{a}_k$ in PMB-SLAM; and both global hypothesis $\mathbf{a}_k$ and vehicle state trajectory $\mathbf{s}_{0:k}$ in marginalized PMB-SLAM).
    
    In the PMB-SLAM filter (see Section~\ref{sec:PMB-SLAM-Filter}), the PMBM is reduced to a PMB~(i.e, result of marginalizing out global hypothesis) after each update step, following the approach from Section~\ref{sec:marginalPMBMethod}, which is achieved by running loopy belief propagation (LBP)~\cite{williams2015marginal} to compute marginal association probabilities.
    Particle weights can no longer be based on expected likelihoods, and a novel approach based on Bethe free energy is proposed. 
    Finally, in the marginalized PMB-SLAM filter (Section~\ref{sec:MarginalizedPMB}), we apply the method from Section~\ref{sec:marginalPMBMethod} and also marginalize out the vehicle state trajectories, allowing for a low-complexity implementation without particles.


    All filters operate according to the prediction and update steps, similar to Section~\ref{sec:TrajecRecur}. In particular, 
    in the update step, we will refer to four steps: \emph{Step i)} missed detections of landmarks that were previously undetected; \emph{Step ii)} newly detected landmarks or clutter that were previously undetected, now representing as MB density; \emph{Step iii)} missed detections of the previously detected landmarks; and \emph{Step iv)} detections from the previously detected landmarks.
    
\subsection{PMBM-SLAM Filter}\label{sec:PMBM-SLAM-Filter}
     We develop the proposed PMBM-SLAM filter from the PMBM-MTT \cite{Garcia-Fernandez2018} {in line with the Bayesian recursion in Section~\ref{sec:TrajecRecur}.}
\subsubsection{Notation}\label{sec:PMBM-SLAM-Filter_Not}
    For notational simplicity, we will denote $f(\mathbf{s}_{0:k}|\mathcal{Z}_{1:k-1}) \triangleq f_{\mathsf{p},k}(\mathbf{s}_{0:k})$, $f(\mathbf{s}_{0:k}|\mathcal{Z}_{1:k})  \triangleq f_{\mathsf{u},k}(\mathbf{s}_{0:k})$, $f(\mathcal{X}|\mathbf{s}_{0:k},\mathcal{Z}_{1:k})\triangleq f_{\mathsf{u},k}(\mathcal{X})$.
    We adopt the RBPF approach for the SLAM filter: using particles $f_{\mathsf{p},k}(\mathbf{s}_{0:k})$ and $f_{\mathsf{u},k}(\mathbf{s}_{0:k})$ are represented as $f_{\mathsf{u},k}(\mathbf{s}_{0:k})\approx \sum_{n=1}^N w_{\mathsf{u},k}^n \delta(\mathbf{s}_{0:k}-\mathbf{s}_{0:k}^n),$ where $w_{\mathsf{u},k}^n \geq 0$ such that $\sum_{n=1}^N w_{\mathsf{u},k}^n=1$, and $f_{\mathsf{u},k}(\mathcal{X})$ is maintained by the set density of the map conditioned on vehicle particles $\mathbf{s}_{0:k}^n,~\forall n$.
    We note that $f(\mathcal{X}|\mathbf{s}_{0:k}^n,\mathcal{Z}_{1:k})$ is a PMBM density consisting of Poisson and MBM densities.
    In undetected landmarks with the Poisson density, we will denote the intensity function conditioned on vehicle sample $n$ by $\lambda(\mathbf{x},m|\mathbf{s}_{0:k}^n,\mathcal{Z}_{1:k})\triangleq \lambda_{\mathsf{u},k}^n(\mathbf{x},m)$, and {we consider the intensity function $\lambda_{\mathsf{u},k}^n(\mathbf{x},m)$ for $m\in\{\text{VA},\text{SP}\}$
    since the BS is regarded as a detected landmark and is thus not considered in the undetected landmarks.}
    In detected landmarks with MBM density, we will denote the MBM density conditioned on vehicle sample $n$ by $f(\mathbf{x},m|\mathbf{s}_{0:k}^n,\mathcal{Z}_{1:k}) \triangleq f_{\mathsf{u},k}^n(\mathbf{x},m)$, and we consider the MBM components $f_{\mathsf{u},k}^n(\mathbf{x},m)$ 
    for $m\in \mathcal{M}$,
    $\beta_{\mathsf{u},k}^n$, and existence probability $r_{\mathsf{u},k}^n$.

    For each particle $n$ in the posterior at time $k-1$, we thus have 
    $\{\mathbf{s}_{0:k-1}^n,w_{\mathsf{u},k-1}^n\}$, PPP
    $\{\lambda_{\mathsf{u},k-1}^n(\mathbf{x},m)\}_{m\in\{\text{VA},\text{SP}\}}$, global hypotheses $\mathbf{a}_{k-1}\in\mathcal{A}_{k-1}^n$, and MBM components $\{\{f_{\mathsf{u},k-1}^{i,a_{k-1}^i,n}(\mathbf{x},m)\}_{m\in\{\text{BS},\text{VA},\text{SP}\}},r_{\mathsf{u},k-1}^{i,a_{k-1}^i,n},\beta_{\mathsf{u},k-1}^{i,a_{k-1}^i,n}\}_{i=1}^{I_{k-1}}$,  
    which will be predicted and updated for all $n$ in the following.
    
\subsubsection{Vehicle Prediction and Map Update}\label{sec:PMBMpred}
    We use~\eqref{eq:mobility} to generate $\mathbf{s}_k^n \sim f(\mathbf{s}_k|\mathbf{s}_{k-1}^n)$, and $w_{\mathsf{p},k}^n=w_{\mathsf{u},k-1}^n$. 
    The map update step consists of standard PMB components update and global hypothesis update from~\cite{Garcia-Fernandez2018}. The details are provided in Appendix~\ref{sec:App-PMBM-update}.

\subsubsection{Vehicle State Update}\label{sec:PMBMupdateVeh}
    {The updated particle weight $w_{\mathsf{u},k}^n$ is related to the prior particle weight $w_{\mathsf{p},k}^n$ through the following proposition. 
\begin{prop}[PMBM-SLAM particle weight update]
    The updated particle weight is given by 
\begin{align}
   w_{\mathsf{u},k}^n \propto &\,w_{\mathsf{p},k}^n \sum_{\mathbf{a}_{k-1}\in \mathcal{A}_{k-1}}\sum_{ \substack{ \uplus_{i=1}^{I_{k-1}} \mathcal{Z}_k^i \uplus \mathcal{Z}_k^{\text{U}}= \mathcal{Z}_k,\\ |\mathcal{Z}_k^i|\le 1}}  \prod_{\mathbf{z}_k^j \in \mathcal{Z}_k^\text{U}}\nu_k^{n}({\{\mathbf{z}_k^j\}})\nonumber\\
    &\times \prod_{i=1}^{I_{k-1}} \nu_k^{i,a_{k-1}^i,n}(\mathcal{Z}_k^i)\beta_{\mathsf{u},k-1}^{i,a_{k-1}^i,n}, 
    \label{eq:ELCompute}
\end{align}
    where $\nu_k^{n}({\{\mathbf{z}_k^j\}})$, $\nu_k^{i,a_{k-1}^i,n}(\mathcal{Z}_k^i)$, and $\beta_{\mathsf{u},k-1}^{i,a_{k-1}^i,n}$ are computed as part of the PMB component update in Appendix~\ref{sec:App-PMBM-map-update}. 
\end{prop}
\begin{proof}
    See Appendix \ref{app:PMBM-weigth-update}.
\end{proof}}

 \begin{rem}[Relation between particle weight and global association weights]
    The global association weights $\beta_k^{\mathbf{a}_k,n}$ for any $\mathbf{a}_k \in \mathcal{A}_{k}^n$ are given by 
\begin{align}
     & \beta_k^{\mathbf{a}_k,n} 
    = \frac{1}{\chi^n_k} \prod_{i=1}^{I_{k}} \beta_{\mathsf{u},k}^{i,a_{k}^i,n} =
    \frac{1}{\chi^n_k} 
    \prod_{i=1}^{I_{k-1}} \beta_{\mathsf{u},k}^{i,a_k^i,n}
    \prod_{i'=I_{k-1}+1}^{I_{k-1}+J_k}\beta_{\mathsf{u},k}^{i',a_k^{i'},n} \notag\\
    & = \frac{1}{\chi^n_k} \prod_{i=1}^{I_{k-1}} \nu_k^{i,a_{k-1}^i,n} (\mathcal{Z}_k^{a_k^i}) \beta_{\mathsf{u},k-1}^{i,a_{k-1}^i,n} \prod_{j \in \text{U}_k(\mathbf{a}_k)} 
    \nu_k^{n}({\{\mathbf{z}_k^j\}}) \label{eq:betaGlobal}
\end{align}
    where (a) $\beta_{\mathsf{u},k}^{i',a_k^{i'},n}=1$ when measurement $j=i'-I_{k-1}$ does not correspond to a newly detected landmark under global hypothesis $\mathbf{a}_k$;
    (b) $\mathcal{Z}_k^{a_k^i}$ is the measurement set (a singleton or empty set) associated with landmark $i$ under global hypothesis $\mathbf{a}_k$, defined as $\mathcal{Z}_k^{a_k^i} = \mathbf{z}_k^{j}$ if $a_k^i=j$, $\mathcal{Z}_k^{a_k^i}=\emptyset$ if $a_k^i=0$;
    and (c) $\text{U}_k(\mathbf{a}_k)$ represents the measurements indices that correspond to newly detected landmarks under global hypothesis $\mathbf{a}_k$. The normalization constant $\chi^n_k$ can be recovered from $\sum_{\mathbf{a}_k}\beta_k^{\mathbf{a}_k,n}=1$. 
    Inspecting~\eqref{eq:ELCompute} and~\eqref{eq:betaGlobal}, we observe that the correction to the particle weight is given by the normalization constant $\chi^n_k$ from \eqref{eq:betaGlobal}.
\end{rem}

\subsection{PMB-SLAM Filter}\label{sec:PMB-SLAM-Filter}
    Here, we develop the PMB-SLAM filter by approximating the PMBM as a PMB, {in line with the Bayesian recursion in Section~\ref{sec:TrajecRecur}.}
    First, similarly to the four steps in Section~\ref{sec:PMBMpred}, we compute Poisson and MB components from the previous PMB density for each particle $n$ and global hypothesis $\mathbf{a}_{k-1}^n\in \mathcal{A}_{k-1}^n$, which will be briefly described in Section~\ref{sec:PMBSLAMAPP}.
    Second, we approximate the PMBM as a PMB of the set of landmarks by using Lemma~\ref{lem:MinKLD}, where the nuisance parameter $\bm{\eta}$ corresponds to the global association hypotheses $\mathbf{a}_k$.
    We note that the density of undetected landmarks does not depend on the nuisance parameter, while each Bernoulli component $i$ only depends on $a_k^i$, i.e., we can express
\begin{align}
     \sum_{\mathbf{a}_k \in \mathcal{A}_k}\tilde{f}^\text{U}(\tilde{\mathcal{X}}^\text{U}|\mathbf{a}_k)\mathsf{p}(\mathbf{a}_k)&=\tilde{f}^\text{U}(\tilde{\mathcal{X}}^\text{U})
     \\
     \sum_{\mathbf{a}_k \in \mathcal{A}_k} \tilde{f}^i(\tilde{\mathcal{X}}^i|\mathbf{a}_k)\mathsf{p}(\mathbf{a}_k)&=\sum_{a_k^i=0}^{J_k}  f^{i,a_k^i}(\tilde{\mathcal{X}}^i)\mathsf{p}(a_k^i),
\end{align}
    where $\mathsf{p}(a_k^i)$ is the marginal association probabilities of Bernoulli $i$.
    We show that it can be implemented by using the marginal association probabilities in TOMB/P~\cite{williams2015marginal}, detailed in~\ref{sec:AppPMBMtoaPMB}.
    Third, we compute the marginal association probabilities and derive that this approximation can be designed by the Bethe approach of free energy~\cite{Yedidia2005BetheFree}. We also derive that the weight of vehicle particle is computed by the marginal belief in the Bethe approach, determined by the marginal association probabilities. 

\subsubsection{Vehicle Prediction and Map Update}\label{sec:PMBSLAMAPP}
    By construction, in a PMB, $\lvert \mathcal{A}_{k-1}^n\rvert=1$.
    Notation that was introduced in Section~\ref{sec:PMBM-SLAM-Filter_Not} is used again in here.
    Hence, for each particle $n$ in the posterior at time $k-1$ we have $\{\mathbf{s}_{0:k-1}^n,w_{\mathsf{u},k-1}^n\}$,
    $\{\lambda_{\mathsf{u},k-1}^n(\mathbf{x},m)\}_{m\in\{\text{VA},\text{SP}\}}$, and $\{\{f_{\mathsf{u},k-1}^{i,n}(\mathbf{x},m)\}_{m\in\{\text{BS},\text{VA},\text{SP}\}},r_{\mathsf{u},k-1}^{i,n}\}_{i=1}^{I_{k-1}}$.
    Vehicle prediction follows Section~\ref{sec:PMBMpred}, and landmark parameters are updated, similar to  Section~\ref{sec:PMBMpred}, ignoring the previous association $a_{k-1}^i$.
    Then, 
    we have $J_k + I_{k-1}(J_k + 1)$ Bernoullis (i.e., $J_k$, $I_{k-1}$, and $I_{k-1}J_k$ Bernoullis are respectively obtained in Steps ii)-iv) of Section~\ref{sec:PMBMpred}), as detailed in Appendix \ref{sec:App-PMBM-map-update}. 
    Note that even the landmark density is a PMB at time $k-1$, at the end of time $k$ the landmark density will be a PMBM. 
    In the next two subsections, we describe how a PMBM is approximated by a PMB and how the particle weights for updating the vehicle density are computed. 
    
    
\subsubsection{Approximating the PMBM as a PMB and Vehicle Update}\label{sec:AppPMBMtoaPMB}
    When the landmark density is a PMB, the global association weights in~\eqref{eq:betaGlobal} become
\begin{align}
    \beta_k^{\mathbf{a}_k,n}  = \frac{1}{\chi^n_k} \prod_{i=1}^{I_{k-1}} \nu_k^{i,n}(\mathcal{Z}_k^{a_k^i}) \prod_{j \in \text{U}_k(\mathbf{a}_k)}
    \nu_k^{n}({\{\mathbf{z}_k^j\}}).
    \label{eq:PMBGlobalWeight}
\end{align}
    From this pmf over global hypotheses, we 
    compute the marginals 
    using the LBP algorithm as in \cite[Appendix C, Fig. 9]{williams2015marginal}. In particular, we introduce an integer random variable 
    $\mathbf{a}_k= [\mathbf{c}_k^\top,\mathbf{d}_k^\top]^\top = [c_k^1,...,c_k^{I_{k-1}},d_k^1,...,d_k^{J_k}]^\top$, where $c_k^i \in \{0,1,\ldots,J_k\}$ and $d_k^j \in \{0,1,\ldots,I_{k-1}\}$. The pmf of $\mathbf{a}_k$ is defined as $\mathsf{p}_k^n(\mathbf{a}_k) =\mathsf{p}_k^n(\mathbf{c}_k,\mathbf{d}_k)$ with
\begin{align}
    \mathsf{p}_k^n(\mathbf{a}_k)& 
    = \frac{1}{Z_k^n} \prod_{i=1}^{I_{k-1}}\prod_{j=1}^{J_{k}}p_{\mathsf{a},k}^n(c_k^i)p_{\mathsf{a},k}^n(d_k^j)\Psi_k^n(c_k^i,d_k^j) \label{eq:DAProbab},
\end{align}
    where
\begin{align}
    p_{\mathsf{a},k}^n(c_k^i=j) & = \begin{cases}\label{eq:to-av}
    \nu_k^{i,n}(\{\mathbf{z}_k^{j}\}),&j\in\{1,...,J_k \},\\
    \nu_k^{i,n}(\emptyset),&j = 0,
    \end{cases}
\end{align}    
\begin{align}    
    p_{\mathsf{a},k}^n(d_k^j=i)& =
    \begin{cases}\label{eq:mo-av}
    1,&i\in\{1,...,I_{k-1}\},\\
    \nu_k^{n}(\{\mathbf{z}_k^{j}\}),&i = 0,
    \end{cases}
\end{align}
\begin{align}
    \Psi_k^n(c_k^i,d_k^j) & = 
    \begin{cases} \label{eq:assocfact}
        0,& c_k^i = j, d_k^j \neq i, \text{ or } \\
        & c_k^i \neq j, d_k^j = i,\\
        1,& \text{otherwise,}
    \end{cases}
\end{align}
    {where $\Psi_k^n(c_k^i,d_k^j)$ ensures that only valid global associations are considered.}
    We observe the correspondence with \eqref{eq:PMBGlobalWeight}, where $\mathsf{p}_k^n(\mathbf{a}_k)=\beta_k^{\mathbf{a},n}$ and $\chi^n_k=Z_k^n$.
    Performing LBP algorithm~\cite[Appendix C, Fig. 9]{williams2015marginal} on the associated factor graph, yields the beliefs $\mathrm{bel}_k^n(c_k^i)$ and $\mathrm{bel}_k^n(d_k^j)$, so that the approximate marginal association probabilities are  
    $\mathsf{p}_k^{i,n}(j) = \mathrm{bel}_k^n(c_k^i=j)$ for $i\in \{1,...,I_{k-1}\}$ and $j\in \{0,...,J_k \}$; and 
    $\mathsf{p}_k^{I_{k-1}+j,n}(0) = \mathrm{bel}(d_k^{j}=0)$ for $j\in \{1,...,J_k \}$.
    From these approximate marginal association probabilities, the approximate PMB can be recovered by the so-called TOMB/P method from \cite{williams2015marginal}, as described in Appendix \ref{app:conversion}.
    
    {
    \begin{prop}[PMB-SLAM particle weight update]
    The updated particle weight is given by
    \begin{align}
         w_{\mathsf{u},k}^n \propto w_{\mathsf{p},k}^n \exp(-\mathsf{F}(\mathrm{bel}_k^n))
    \end{align}
    where $\mathsf{F}(\mathrm{bel}_k^n)$ is the so-called Bethe free energy of the beliefs~\cite{Yedidia2005BetheFree}. 
    \end{prop}
    \begin{proof}
        See Appendix \ref{app:PMB-weigth-update}
    \end{proof}
    }
This proposition thus provides a tractable way to compute the particle weights.

\subsection{Marginalized PMB-SLAM Filter}\label{sec:MarginalizedPMB}
    We develop the proposed marginalized PMB-SLAM filter {in line with the Bayesian recursion in Section~\ref{sec:StateRecur}.}
    A marginalized PMB-SLAM filter is derived by marginalizing joint posterior density for a vehicle state and landmark $f(\mathbf{s}_k,\mathcal{X}|\mathcal{Z}_{1:k})$.
    When computing the posterior density of the set of landmarks $f(\mathcal{X}|\mathcal{Z}_{1:k})$,
    we approximate the PMBM as a PMB of the set of landmarks by using Lemma~\ref{lem:MinKLD}. Here, the nuisance parameter $\bm{\eta}$ corresponds to the global association hypotheses $\mathbf{a}_k$ and vehicle state $\mathbf{s}_k$. 
    
\subsubsection{Notation}
    For notational convenience,
    we will denote $f(\mathbf{s}_{k}|\mathcal{Z}_{1:k-1}) \triangleq f_{\mathsf{p},k}(\mathbf{s}_{k})$, $f(\mathbf{s}_{k}|\mathcal{Z}_{1:k}) \triangleq f_{\mathsf{u},k}(\mathbf{s}_{k})$, 
    $\lambda(\mathbf{x},m|\mathcal{Z}_{1:k}) \triangleq\lambda_{\mathsf{u},k}(\mathbf{x},m)$, 
    $f^i(\mathbf{x},m|\mathcal{Z}_{1:k}) \triangleq f_{\mathsf{u},k}^i(\mathbf{x},m)$,
    and $f(\mathcal{X}|\mathcal{Z}_{1:k}) \triangleq f_{\mathsf{u},k}(\mathcal{X})$.
    The purpose is to estimate $f_{\mathsf{u},k}(\mathbf{s}_k)$ of \eqref{eq:UpMarVeh} and $f_{\mathsf{u},k}(\mathcal{X})$ of \eqref{eq:UpMarObj}, and for which
    we adopt the RFS-likelihood function in \cite[eq. (25), (26)]{Garcia-Fernandez2018} and PMB representation in \eqref{eq:PMBdensity}, which are respectively plugged into  $g(\mathcal{Z}_k|\mathbf{s}_k,\mathcal{X})$ and $f_{\mathsf{u},k-1}(\mathcal{X})$.

    In the posterior at time $k-1$,   $\{\lambda_{\mathsf{u},k-1}(\mathbf{x},m)\}_{m\in\{\text{VA},\text{SP}\}}$, $f_{\mathsf{u},{k-1}}(\mathbf{s}_{k-1})$, and $\{\{f_{\mathsf{u},k-1}^{i}(\mathbf{x},m)\}_{m\in\{\text{BS},\text{VA},\text{SP}\}},r_{\mathsf{u},k-1}^{i}\}_{i=1}^{I_{k-1}}$ are given, which will be predicted and updated at time $k$.
        
\subsubsection{Vehice Prediction and Map Update}\label{sec:marginal_recursion}
    The vehicle density is predicted:
\begin{align}
    f_{\mathsf{p},k}(\mathbf{s}_k) = \int f(\mathbf{s}_{k}|\mathbf{s}_{k-1}) f_{\mathsf{u},k}(\mathbf{s}_{k-1}) \mathrm{d} \mathbf{s}_{k-1}. \label{eq:PredMarVeh}
\end{align}
    {The updated density of set of landmarks is represented as PMB components.
    Thus, we also use Step i)--iv) for updating Poisson and MB components, which were introduced in Section~\ref{sec:PMBintution} and were presented in Section~\ref{sec:PMBMpred}. 
    However, instead of conditioning on a particle $n$, we marginalize out the vehicle state.} 
    {See Appendix~\ref{app:MPMB-veh}.}

\subsubsection{Approximating PMBM as a PMB and Vehicle Update}\label{sec:marginal_PMBMasPMB}

    
    Similarly to Section~\ref{sec:AppPMBMtoaPMB}, we compute the approximate marginal association probabilities (i.e., $\mathsf{p}_k^{i}(j) = \mathrm{bel}(c^i=j)$ for $i\in \{1,...,I_{k-1}\}$ and $j\in \{0,...,J_k \}$; and 
    $\mathsf{p}_k^{I_{k-1}+j}(0) = \mathrm{bel}(d^{j}=0)$ for $j\in \{1,...,J_k \}$).
    Then, we compute the density and existence probability of landmarks for the PMB, similarly to Appendix \ref{app:conversion}.
\begin{prop}
    {The vehicle posterior density is proportionally approximated as a mixture density
\begin{align}
    f_{\mathsf{u},k}(\mathbf{s}_k)
    & \appropto 
    \sum_{\uplus_{i=1}^{I_{k-1}}\mathcal{Z}_k^i \uplus \mathcal{Z}_k^\text{U} =\mathcal{Z}_k} 
    \prod_{\mathbf{z} \in \mathcal{Z}_k^\text{U}} \nu_k(\{\mathbf{z}\}) \prod_{i=1}^{I_{k-1}} \nu_k^i(\mathcal{Z}_k^i) \nonumber\\
    &\times q(\mathbf{s}_k|\mathcal{Z}_k^\text{U},\mathcal{Z}_k^1,...,\mathcal{Z}_k^{I_{k-1}}),
    \label{eq:MarVehPosApp}
\end{align}
where $q(\mathbf{s}_k|\mathcal{Z}_k^\text{U},\mathcal{Z}_k^1,...,\mathcal{Z}_k^{I_{k-1}})$ is a normalized density, describing the posterior density conditioned on the association $\mathcal{Z}_k^\text{U},\mathcal{Z}_k^1,...,\mathcal{Z}_k^{I_{k-1}}$. The constants $\nu_k(\{\mathbf{z}\})$, $\nu_k^i(\emptyset)$, and $\nu_k^i(\{\mathbf{z}\})$ were determined in Appendix~\ref{app:MPMB-veh}.
}
{We can express \eqref{eq:MarVehPosApp} 
as 
\begin{align}
    & f_{\mathsf{u},k}(\mathbf{s}_k) \approx \label{eq:MarVehPosApp2}  \\
    &  \sum_{\mathbf{a}_k} \mathsf{p}_k(\mathbf{a}_k)  f_{\mathsf{p},k}(\mathbf{s}_k) \prod_{j \in \text{U}_k(\mathbf{a}_k)}\frac{\psi_k(\mathbf{z}_j,\mathbf{s}_k)}{\nu_k(\{\mathbf{z}_j\})}\prod_{i=1}^{I_{k-1}} \frac{ q_i(\mathcal{Z}_k^{a_k^i}|\mathbf{s}_k)}{\nu_k^i(\mathcal{Z}_k^{a^i_k})}, \notag 
\end{align}
where $q_i(\mathbf{z}_j|\mathbf{s}_k)$ and $q_i(\emptyset|\mathbf{s}_k)$ are defined in Appendix \ref{app:MPMB-target}. 
    }
\end{prop}
\begin{proof}
    {See Appendix~\ref{app:MPMB-target}}
\end{proof}

    As the nuisance parameters are marginalized out, the correlation between the vehicle state and the landmarks is lost. However, the marginal densities allow for efficient representation and computation, which will be implemented in Section~\ref{sec:Implementation}.

\subsection{{Connection of BP-SLAM to Marginalized PMB-SLAM}}\label{sec:BPtoMPMB}
    {We discuss connection of the BP-SLAM filter~\cite{Erik_BPSLAM_TWC2019,Erik_AOABPSLAM_ICC2019} to the proposed marginalized PMB-SLAM filter, covering the data associations, steps i)--iv) in the PMB update, and also the vehicle posterior computation.}
    
\subsubsection{{Data Association}}
    {Both BP-SLAM and marginalized PMB-SLAM use belief propagation to compute marginal data association probabilities.  The iterative data association step in \cite[Sec.~V-B3]{Erik_BPSLAM_TWC2019} is identical to performing LBP on~\eqref{eq:DAProbab}. Hence, the corresponding beliefs are identical in both methods. However, internally BP-SLAM uses message in the computation, not beliefs.}

\subsubsection{{PMB Component Update}}
    {In BP-SLAM, the landmarks are separated into undetected features (similar to our Step i)), new potential features~(PFs) (similar to our Step ii)), and legacy PFs (similar to our Step iii) and iv)). 
    When we adopt the likelihood functions eqs.~\eqref{eq:likelihood_Poi} and~\eqref{eq:likelihoodBer} with the association variables $c_k^i$ and $d_k^j$ of~\eqref{eq:DAProbab} instead of the likelihood functions in~\cite[Sec.~III-D]{Erik_BPSLAM_TWC2019}, the PF beliefs on the factor graph \cite[Fig.~2]{Erik_BPSLAM_TWC2019} are identical to the proposed marginalized PMB implementation.
    The connections are further detailed as follows:
    \begin{itemize}
        \item Step i) There is no connected message passing step in~\cite{Erik_BPSLAM_TWC2019,Erik_AOABPSLAM_ICC2019}.
        As an ad-hoc modification, the PHD intensity is adopted outside of the factor graph framework and posterior density expression.
        \item Step ii) The message from likelihood function for new PF to the association variable is identical to $\nu_k(\{\mathbf{z}_k^j\})$ of Step ii). The belief for new PF, i.e.,~\cite[eq.~(35)]{Erik_BPSLAM_TWC2019} can be represented as the existence probability~\eqref{eq:TOMBP_NLr} and Bernoulli density~\eqref{eq:TOMBP_NLf}. 
        \item Step iii) and iv) The message from likelihood function for legacy PF to the association variable,~i.e.,~\cite[eq.~(25)]{Erik_BPSLAM_TWC2019},  is identical to $\nu_k^i(\emptyset)$ of Step iii) with $c_k^i=0$, and to $\nu_k^i(\{\mathbf{z}_k^j\})$ of Step iv) with $c_k^i=j$~(see~\eqref{eq:to-av}). After normalization, the belief for legacy PF, i.e.,~\cite[eq.~33]{Erik_BPSLAM_TWC2019} can be represented as the existence probability~\eqref{eq:TOMBP_ELr} and Bernoulli density~\eqref{eq:TOMBP_ELf}.
    \end{itemize}
    It follows that the landmark update in BP-SLAM is identical to marginalized PMB-SLAM, though with a slightly different likelihood function and with an ad-hoc version of Step i).}

    
\subsubsection{{Vehicle State Update}}
    {
    The belief for the vehicle, i.e.,~\cite[eq.~37]{Erik_BPSLAM_TWC2019}, is computed with the  association messages, averaging different local associations \cite[eq.~32]{Erik_BPSLAM_TWC2019} and integrating out the landmark state. Messages over the vehicle state from each landmark are then multiplied with the prior, leading to the vehicle posterior belief. The vehicle state update in marginalized PMB \eqref{eq:MarVehPosApp2}  can be expressed in a similar form,  but has an additional factor $\psi_k(\mathbf{z}_j,\mathbf{s}_k)$ that accounts for the undetected landmarks.}

\section{Implementation of the Marginalized PMB-SLAM Filter}
\label{sec:Implementation}
    
    The implementation of the PMBM-SLAM and PMB-SLAM filters is standard, due to the conditioning of the landmark state on the vehicle state. Hence, a standard representation and implementation can be used \cite{williams2015marginal}. In the marginalized PMB-SLAM, on the other hand, the landmark is not conditioned on the vehicle state and thus requires a novel implementation. We chose a CKF-based implementation, due to the nonlinear nature of the measurements \eqref{eq:measurementModel}, where both the vehicle state and Bernoulli components in the PMB are represented by Gaussian densities.

\subsubsection{{Initialization}}
    {We set  $f_{\mathsf{u},k}(\mathbf{s}_{0})=\mathcal{N}(\mathbf{s};\mathbf{s}_{\mathsf{u},0},\mathbf{U}_{\mathsf{u},0})$.
    We generate cubature points (CPs)\footnote{
     We define $\mathsf{d}{(\star)}$ as the number of unknown variables of $\star$, and
    we define $\epsilon^c_{\star}$ as the $c$-th column vector of the matrix $\mathbf{C}\in \mathbb{R}^{\mathsf{d}{(\star)}\times 2\mathsf{d}{(\star)}}$, computed as $\mathbf{C}=\sqrt{\mathsf{d}(\star)}[\mathbf{I}^{\mathsf{d}{(\star)}},-\mathbf{I}^{\mathsf{d}{(\star)}}]$.  $\sqrt{\mathsf{d}(\star)}[\mathbf{I}^{\mathsf{d}{(\star)}},-\mathbf{I}^{\mathsf{d}{(\star)}}]\in \mathbb{R}^{\mathsf{d}{(\star)}\times 2\mathsf{d}{(\star)}}$, where $\mathbf{I}^{\mathsf{d}{(\star)}}\in \mathbb{R}^{\mathsf{d}{(\star)}\times \mathsf{d}{(\star)}}$ is the identity matrix.} and weights with $\mathbf{U}_{\mathsf{u},0} = \mathbf{C}\mathbf{C}^\top$, $\{\mathbf{s}_{0}^c,w_{\mathsf{u},0}^c\}_{c=1}^{2\mathsf{d}(\mathbf{s}_0)}$, where $\mathbf{s}_{0}^c=\mathbf{C}\epsilon^{c}_{\mathbf{s}_0} + \mathbf{s}_{\mathsf{u},0}$, 
    and weights $w_{\mathsf{u},0}^c=1/(2\mathsf{d}(\mathbf{s}_0))$.
    We have 
    $\{\lambda_{\mathsf{u},0}(\mathbf{x},m)=\kappa(m) \mathcal{U}(\mathbf{x})\}_{m\in\{\text{VA},\text{SP}\}}$,
    $f_{\mathsf{u},0}^{1}(\mathbf{x},\text{BS})=\mathcal{N}(\mathbf{x};\mathbf{x}_\text{BS},\mathbf{I}^{\mathsf{d}(\mathbf{x})})$, $\{f_{\mathsf{u},0}^{1}(\mathbf{x},m)=0\}_{m\in\{\text{VA},\text{SP}\}}$, $r_{\mathsf{u},0}^{1}=1$.}
    
    \subsubsection{Vehicle Density Prediction}
    To implement \eqref{eq:PredMarVeh}, from $\mathcal{N}(\mathbf{s}_{k-1};\mathbf{s}_{\textsf{u},k-1},\mathbf{U}_{\textsf{u},k-1})$ we decompose $\mathbf{U}_{\mathsf{u},k-1} = \mathbf{C}\mathbf{C}^\top$ and generate CPs for $c=1,...,2\mathsf{d}(\mathbf{s}_k)$: $\mathbf{s}_{\mathsf{u},k-1}^{c} = \mathbf{C}\epsilon^{c}_{\mathbf{s}_k} + \mathbf{s}_{\mathsf{u},k-1}$,
    and weights $w_{\mathsf{u},k-1}^c=1/(2\mathsf{d}(\mathbf{s}_k))$.
    We propagate the CPs as ${\mathbf{s}}^c_{k} = \mathsf{v}(\mathbf{s}^c_{\mathsf{u},k-1})$ and $w_k^c = w_{\mathsf{u},k-1}^c$ for all $c$, and we compute the predicted vehicle density $f_{\mathsf{p},k}(\mathbf{s}_{k})=\mathcal{N}(\mathbf{s}_k;\mathbf{s}_{\mathsf{p},k},\mathbf{U}_{\mathsf{p},k})$, where $\mathbf{s}_{\mathsf{p},k} =
    \sum_{c=1}^{2\mathsf{d}(\mathbf{s}_k)}w_{k}^c{\mathbf{s}}^c_{k}$ and $\mathbf{U}_{\mathsf{p},k}=\sum_{c=1}^{2\mathsf{d}(\mathbf{s}_k)}w_{k}^c{\mathbf{s}}^c_{k}{\mathbf{s}}_{k}^{c\top} - \mathbf{s}_{\mathsf{p},k} \mathbf{s}_{\mathsf{p},k}^{\top}+\mathbf{Q}.$

\subsubsection{PMB Update}
The PMB update requires the CKF implementation of the update steps i)--iv), marginalizing out the vehicle state. The details are provided in Appendix \ref{app:MPMB-mapCKF}.

\subsubsection{Vehicle Posterior Computation}
To compute the posterior, we first of all approximate $\psi_k(\mathbf{z}_j,\mathbf{s}_k)$ as a constant in $\mathbf{s}_k$, as measurements related to undetected targets provide limited information regarding the vehicle state. Similarly, $q_i(\emptyset|\mathbf{s}_k)$ is approximated as constant.  
Secondly, to avoid a complex mixture density for 
in $f_{\mathsf{u},k}(\mathbf{s}_k)$, we limit the summation in \eqref{eq:MarVehPosApp2} to the most likely association (determined via Murthy's algorithm or by making hard decision based on the marginal data association belief) and by considering only landmarks for which  $\mathsf{p}_{\mathsf{u},k-1}^i(m^i) > T_\text{EP}$. Then, 
\begin{align}
    f_{\mathsf{u},k}(\mathbf{s}_k)  \appropto & \mathcal{N}(\mathbf{s}_k;\mathbf{s}_{\mathsf{p},k},\mathbf{U}_{\mathsf{p},k}) \prod_{i\in \mathcal{I}_k^*} \int \mathcal{N}(\mathbf{z}_k^j;\mathsf{h}(\mathbf{s}_{k},\mathbf{x},m),{\mathbf{R}_k^j}) \nonumber\\
    &
    \times \mathcal{N}(\mathbf{x}; \mathbf{x}_{\mathsf{u},k-1}^i(m^i), \mathbf{P}_{\mathsf{u},k-1}^i(m^i))\mathrm{d} \mathbf{x} \\
    = & \int f([\mathbf{s}_k^\top,{\mathbf{y}_k}^\top]^\top|\mathcal{Z}_{1:k})\mathrm{d}\mathbf{y}_k, \label{eq:jointUpdate}
\end{align}
    where we utilize $\lvert \mathcal{I}_k^*\rvert$ determined Bernoulli with $\mathsf{p}_{\mathsf{u},k-1}^i(m^i) > T_\text{EP}$, $\mathbf{y}_k=[{\mathbf{x}_{\mathsf{u},k-1}^{\mathcal{I}_k^*(1)}}^\top,...,{\mathbf{x}_{\mathsf{u},k-1}^{\mathcal{I}_k^*(\lvert \mathcal{I}_k^* \rvert)}}^\top]^\top$
    and $f([\mathbf{s}_k^\top,{\mathbf{y}_k}^\top]^\top|\mathcal{Z}_{1:k})$ of~\eqref{eq:jointUpdate} is computed by the CKF.
\begin{figure*}[t!]
\begin{centering}
	\subfloat[at time $k=3$ \label{Fig:T_2}]
	{\includegraphics[width=.65\columnwidth]{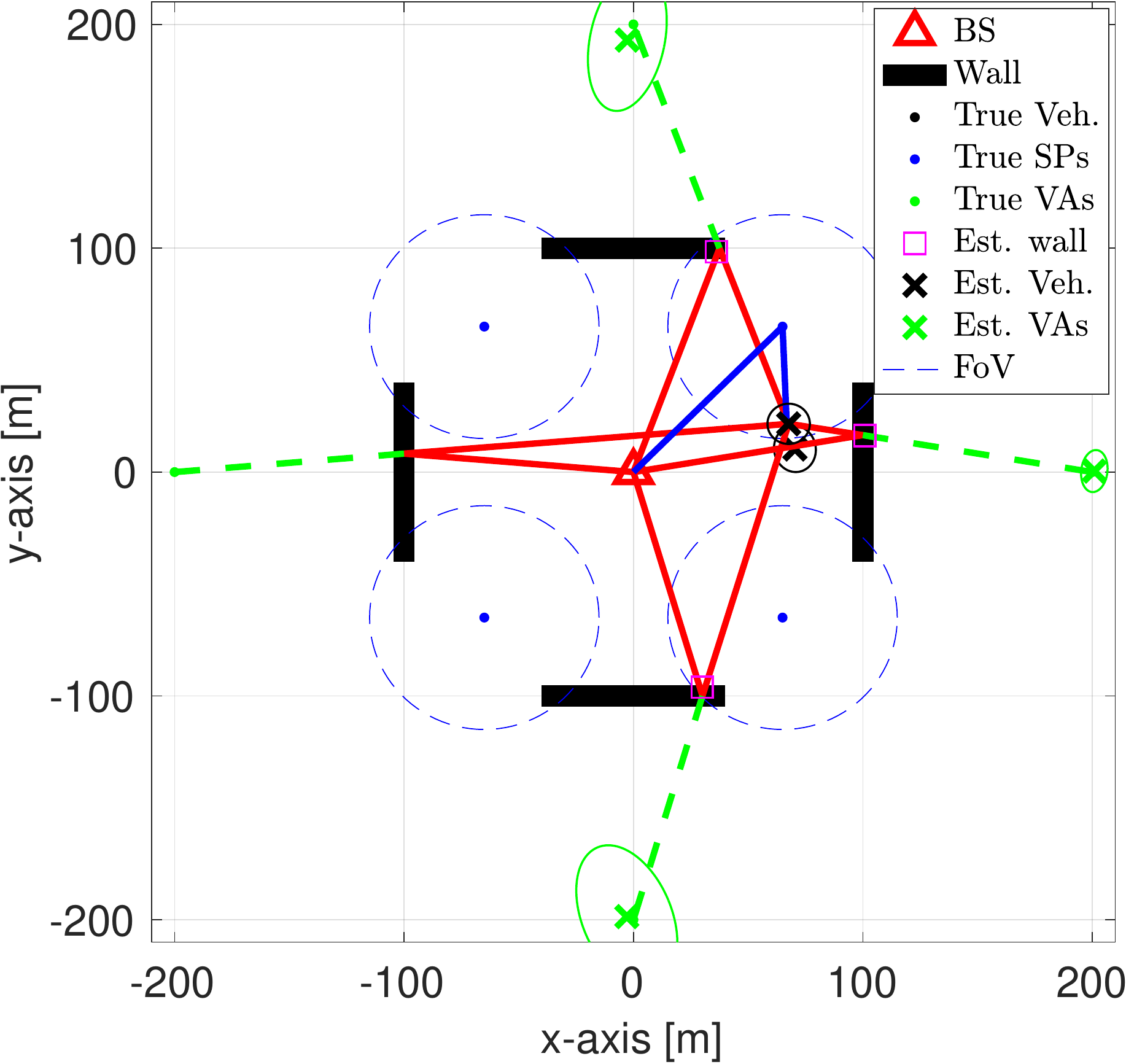}}~~
	\subfloat[at time $k=7$ \label{Fig:T_7}]
	{\includegraphics[width=.65\columnwidth]{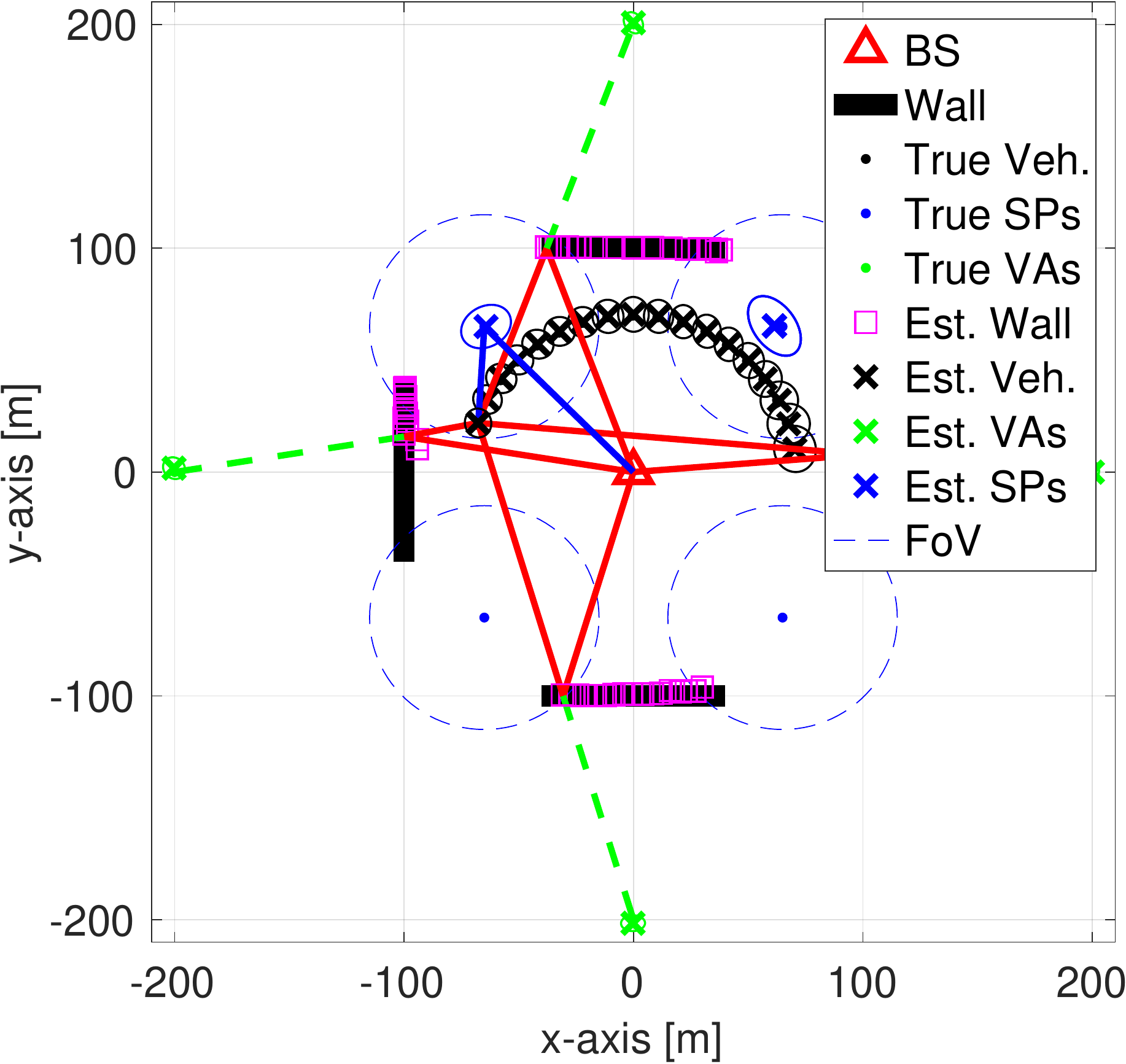}}~~
	\subfloat[at time $k=37$ \label{Fig:T_37}]
	{\includegraphics[width=.65\columnwidth]{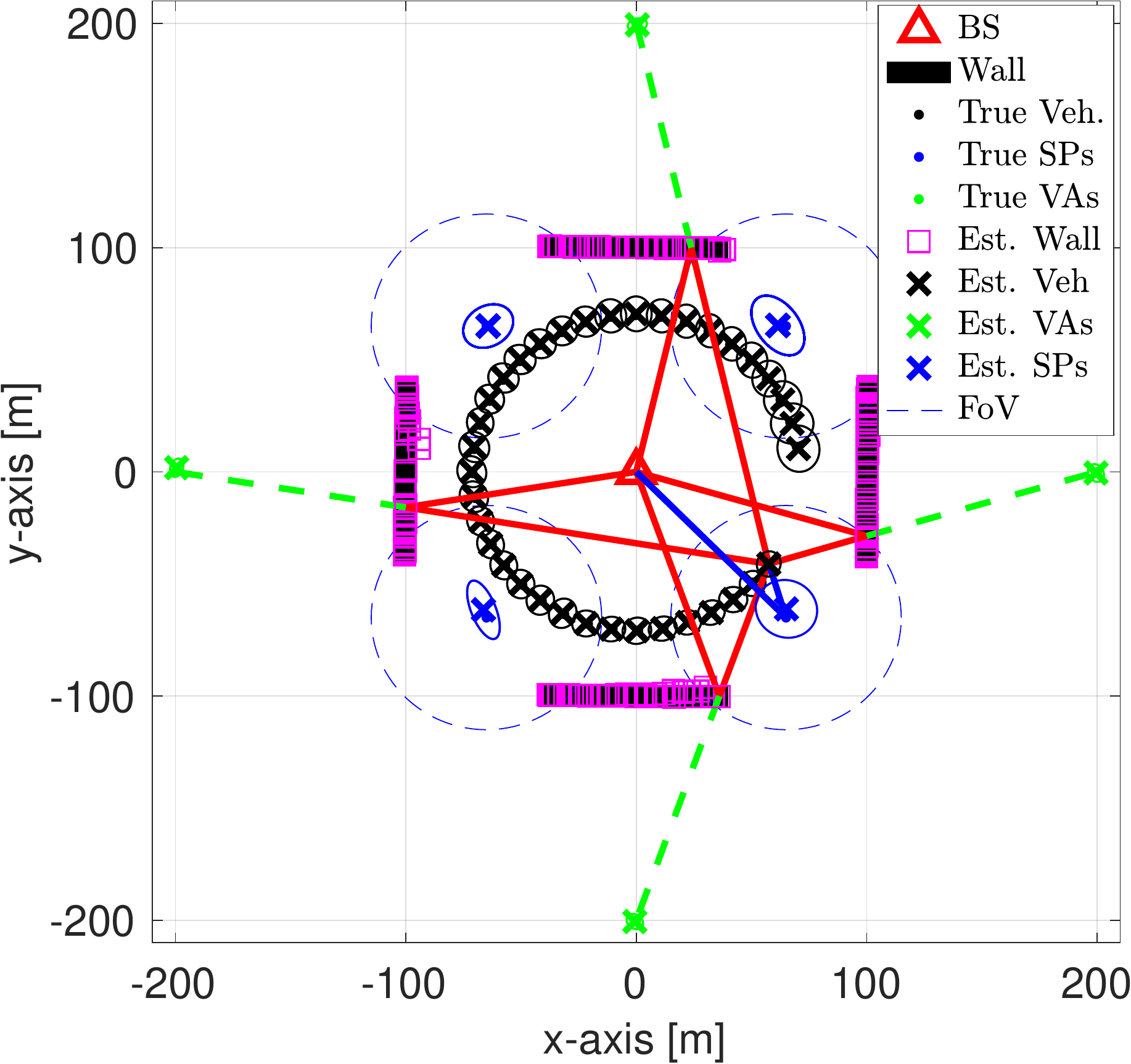}}
	\caption{{Exemplary results of marginalized PMB-SLAM filter at time (a) $k=3$, (b) $k=7$, and (c) $k=37$ in the 5G mmWave vehicular network of Section~\ref{sec:SimulSet}.
	The estimated results of the vehicle and landmarks are represented by x and y location elements of the filtered Gaussians, and filtered mean and covariance are respectively represented by `x' markers and solid ellipses.
	The visible signal paths are represented by the solid line.
	}}
	\label{Fig:SI}
	\par
\end{centering}
\end{figure*}

\section{Numerical Results}\label{sec:Results}
In this section, simulation setup for evaluating the proposed SLAM filters is introduced. Subsequently, performance and results are discussed. {As our focus is on the relative comparison between the different PMBM-based  filters, no extensive evaluation against methods based on other set densities is conducted. The interested reader is referred to~\cite{Garcia-Fernandez2018} for comparison with cardinalized probability hypothesis density~(CPHD), to~\cite{Garcia_Trajectory_2020} for comparison with $\delta$-GLMB and LMB, and to~\cite{meyer2018message} for comparison with BP. 
}

\subsection{Simulation Setup}\label{sec:SimulSet}
   To show the efficiency of the proposed SLAM filters, we evaluate and discuss the performance.
    We consider 3D vehicular networks~(see, Fig.~\ref{Fig:SI}), where a single vehicle is moving with the dynamics~\eqref{eq:mobility}, following the mobility model \cite[Chapter 5]{thrun2005probabilistic}, and fixed landmarks~(a single BS, four VAs, and four SPs) are located.
    The measurements can be obtained with the detection probability within the FoV, including clutter.
    For the details of the network size, initial vehicle state, mobility model, locations of landmarks, FoV, clutter intensity, and performance metrics, we adopt the same definition and values from~\cite[Sec.~VI-A]{Hyowon_TWC2020}.

    We consider BS and vehicle to be equipped with uniform planar arrays, and the number of antennas at the BS and vehicle is respectively 64 ($8\times 8 $) and 16 ($4\times 4$). Each array is equipped with 1 radio-frequency chain, so that analog beamforming is utilized. 
    The  carrier frequency is set to 28~GHz, while  the transmitted signal and noise power spectral density are respectively 5~dBm and $-$174~dBm/Hz.
    We consider OFDM pilot signals with 64 subcarriers in 200~MHz bandwidth and 16 OFDM symbols. The beamforming weights are set randomly during each OFDM symbol. To derive the measurement noise covariance $\mathbf{R}_k^j$, we consider the 5G mmWave specific features and adopt the Fisher information matrix~(FIM) of channel parameters~\cite{Zohair_5GFIM_TWC2018}.
   To implement the proposed SLAM filters, the following details are considered. We implement the intensity density for undetected landmarks by the uniform density $\mathcal{U}(\mathbf{x})$ with the weight $\kappa(m)$: $\lambda(\mathbf{x},m)=\kappa(m)\mathcal{U}(\mathbf{x})$ and set $\kappa(m) = 2.37\times 10^{-6}$ for $m = \{\text{VA},\text{SP}\}$.
    The adaptive detection probability $\mathsf{p}(\mathbf{s}_k,\mathbf{x},m)$ is computed per MB and landmark type. We set $\mathsf{p}(\mathbf{s}_k,\mathbf{x},m)~\approx 0. $ in Step i), $\mathsf{p}(\mathbf{s}_k,\mathbf{x},m)\approx P_\text{D}$ in Step ii). In Step iii), we set $\mathsf{p}(\mathbf{s}_k,\mathbf{x},m)\approx P_\text{D}$ if $M_\text{d}<T_\text{G}$, otherwise, $\mathsf{p}(\mathbf{s}_k,\mathbf{x},m)~\approx 0$, where $M_\text{d}$ is ellipsoidal gating distance and $T_\text{G}$ is the gating threshold, computed by using the Chi-square CDF with gate probability $P_\text{G}$. We set $P_{\text{D}}=0.95$ and $P_\text{G}=0.99$.
    We use a CKF approximation to compute the landmark posterior densities.
    For numerical robustness in the CKF update for Step ii) and iv), 
    we replace ${\mathbf{R}_k^j}$ with  ${\mathbf{R}_{\mathsf{U},k}^j=4\times \mathbf{R}_k^j}$.
    {Due to the nonlinear measurement, the updated covariance of vehicle state is excessively concentrated in the CKF. To mitigate this adverse effect, the dithering method is adopted~\cite{Gustafsson_EKF_ICCASSP2008}.}
    In data association, we set the maximum allowable number of global hypotheses $B_\text{max}=200$~\cite{Garcia-Fernandez2018} and the number of iterations for computing the marginal association probabilities $L_\text{max} = 100$~\cite{Williams_MAPBP_TAES2014}.
    After the update step at every time $k$, the pruning and merging step is performed, which consists of the following steps. 
    A Bernoulli is pruned when its existence probability is less than $10^{-5}$,
    and a global hypothesis is pruned when its weight is less than $10^{-4}$.
    
    {A landmark is detected when the existence probability of Bernoulli is larger than $T_\text{EP}=0.4$, and its landmark type is determined as $m^* = \max_{m} e(m)$.}
    Simulation results were obtained by averaging over 10 Monte Carlo runs, and by using
    MATLAB implementation for the proposed three SLAM methods, which were executed on a PC with 3.2 GHz Intel Core i7-8700 CPU and 32 GB RAM, and the operating system is Windows 10 Pro 64-bit.

\begin{figure}
\begin{centering}
	\subfloat[\label{Fig:VA_GOSPA}]
	{\includegraphics[width=.95\columnwidth]{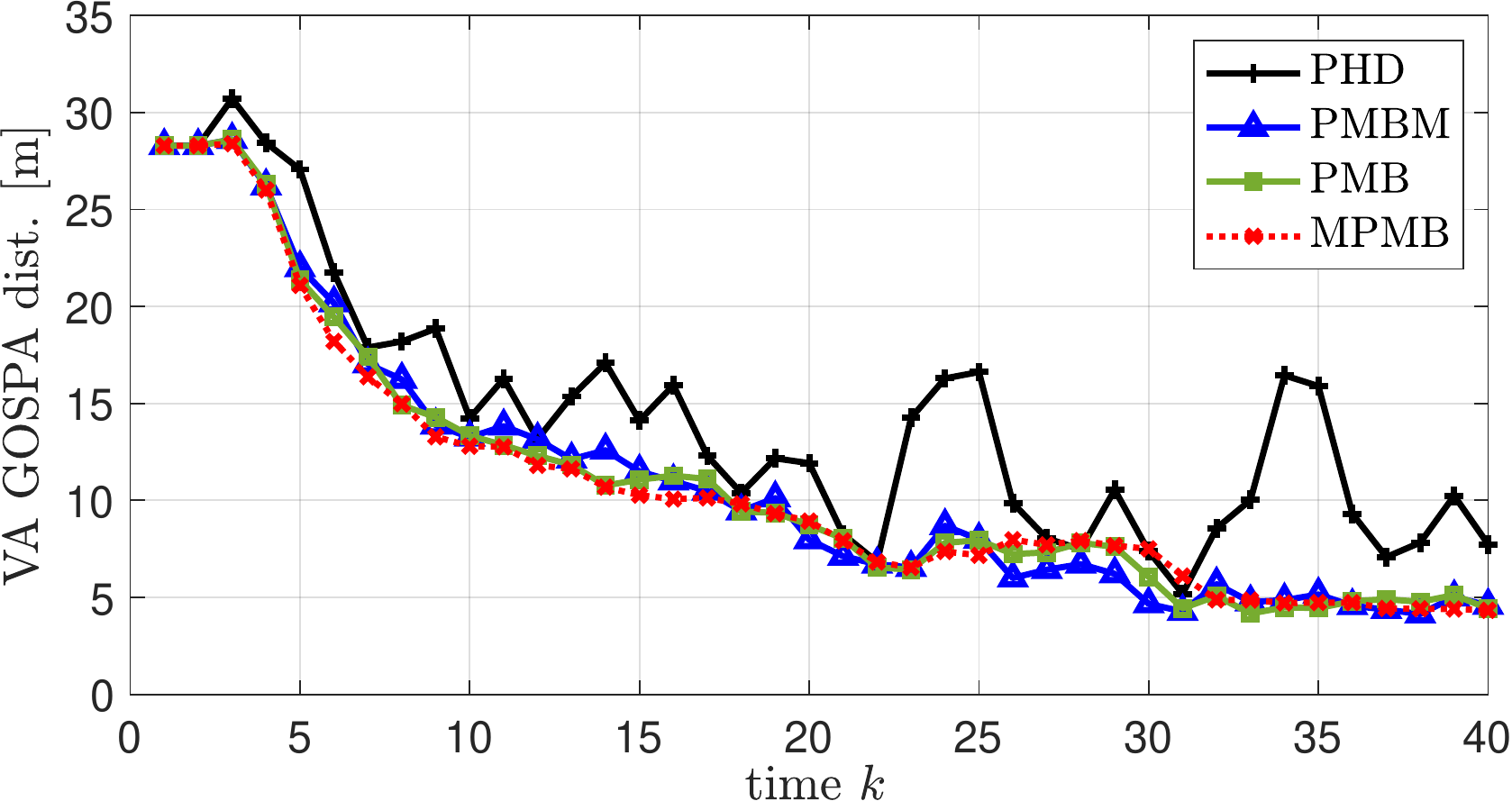}}\hfill
	\subfloat[\label{Fig:SP_GOSPA}]
	{\includegraphics[width=.95\columnwidth]{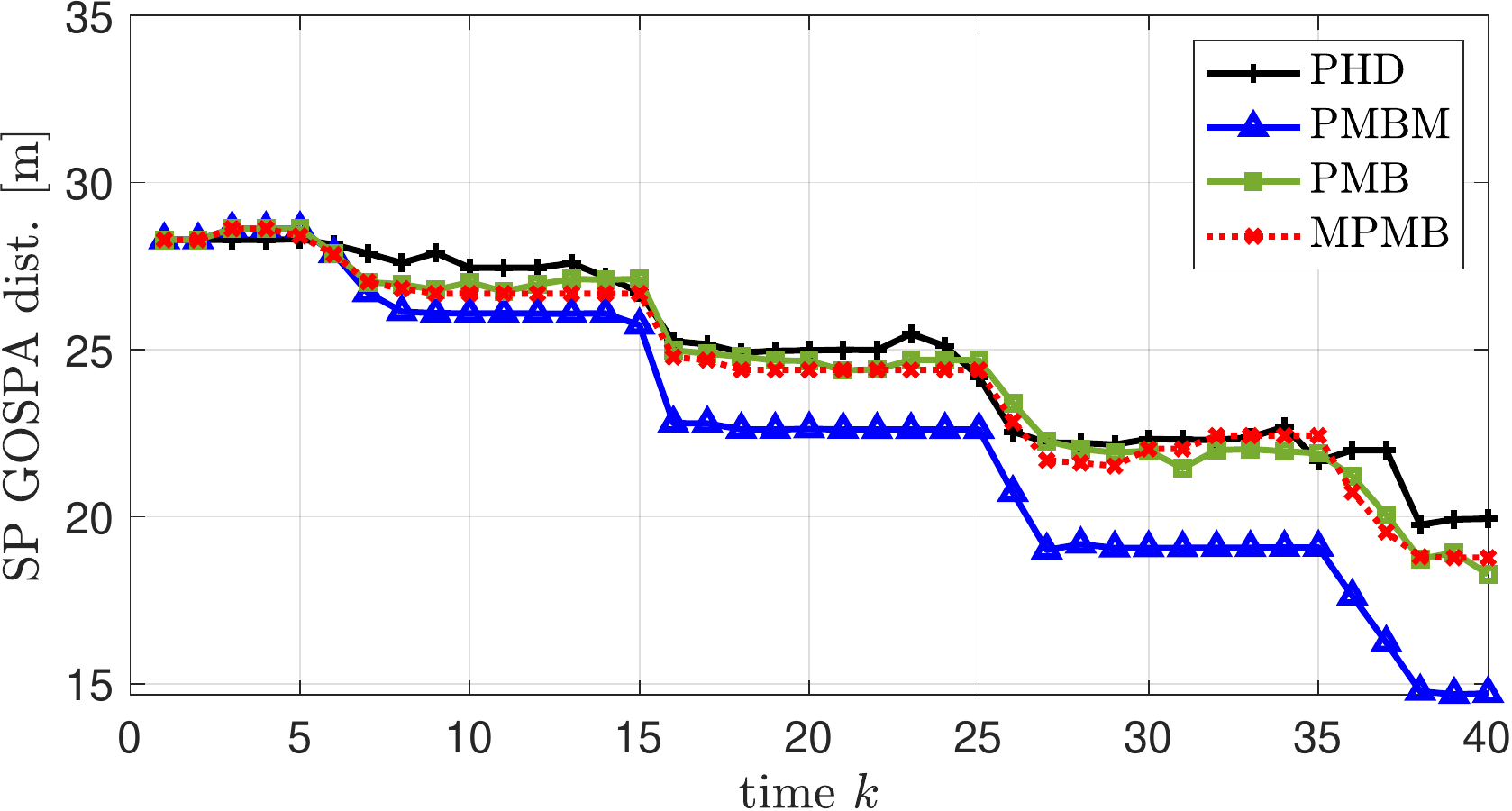}}
	\caption{Comparison of the average GOSPA of (a) VA and (b) SP for the three proposed SLAM filters (PMBM, PMB, marginalized PMB) and PHD-SLAM~\cite{Hyowon_TWC2020}.
	}
	\label{Fig:Mapping}
	\par
\end{centering}
\end{figure}


\begin{figure}
\begin{centering}
	\subfloat[\label{Fig:Loc_RMSE}]
	{\includegraphics[width=.95\columnwidth]{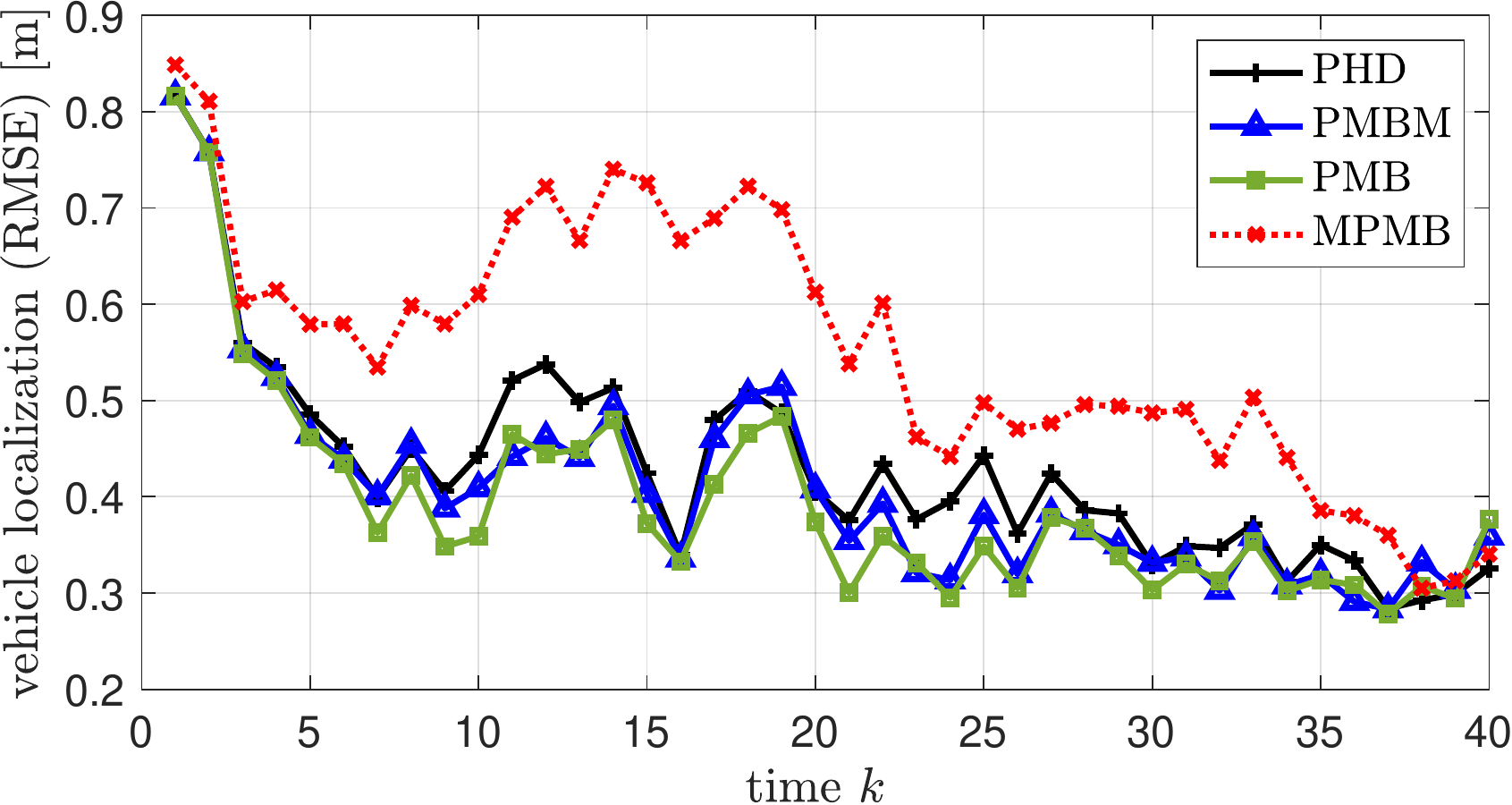}}\hfill
	\subfloat[\label{Fig:Bias_RMSE}]
	{\includegraphics[width=.95\columnwidth]{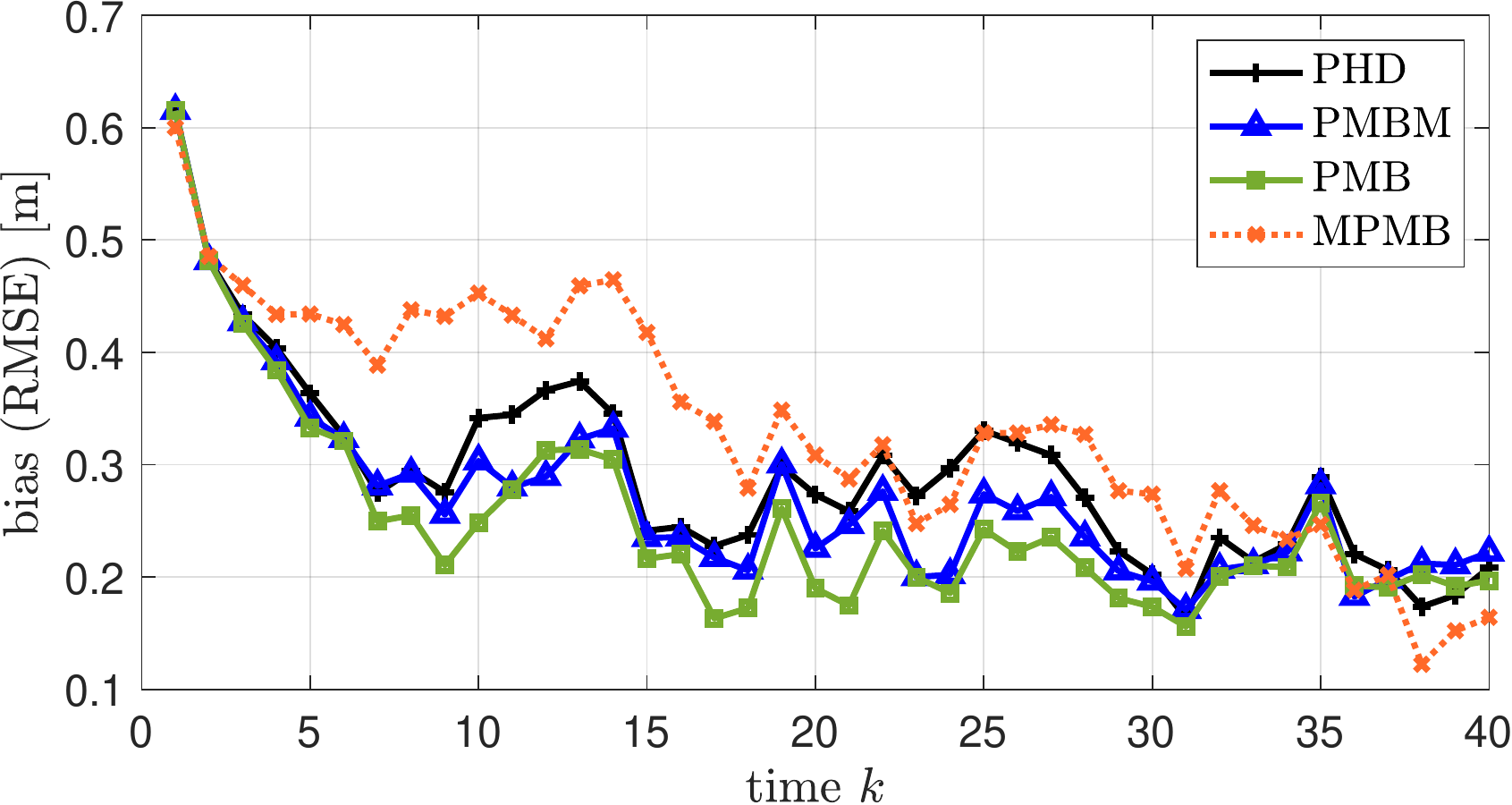}}\hfill
	\subfloat[\label{Fig:Heading_RMSE}]
	{\includegraphics[width=.95\columnwidth]{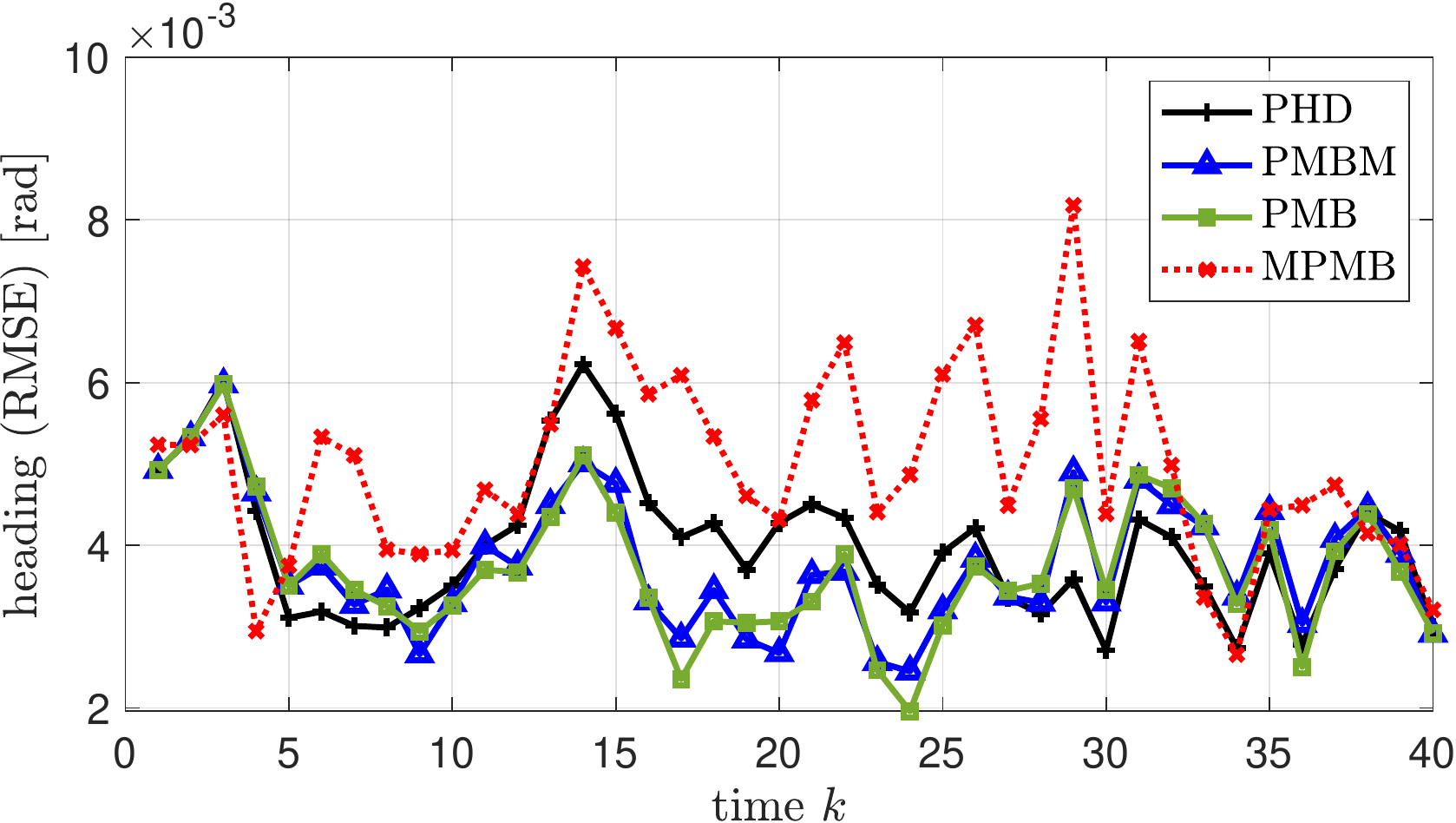}}
	\caption{Comparison of the average RMSE of (a) location, (b) bias, and (c) heading for the three proposed SLAM filters (PMBM, PMB, and marginalized PMB) and PHD-SLAM \cite{Hyowon_TWC2020}.
	}
	\label{Fig:Positioning}
	\par
\end{centering}
\end{figure}

\subsection{Results and Discussions}
    Here, we demonstrate the performance of the proposed three SLAM filters with application to 5G radio-SLAM, compared to the PHD-SLAM filter~\cite{Hyowon_TWC2020} as a benchmark.
    Furthermore, exemplary results of the proposed marginalized PMB-SLAM filter is provided in Fig.~\ref{Fig:SI}, showing the vehicle density and the SP and VA estimates. {The estimated walls are also shown,
    most clearly seen in Fig.~\ref{Fig:SI}c. The walls are directly estimated based on the estimated VAs, so that VA estimation performance will be used to indicate how well the walls can be estimated~\cite[eq.~(42)]{Hyowon_TWC2020}.}
\subsubsection{Mapping}
    Fig.~\ref{Fig:Mapping} shows the mapping performance of the three proposed SLAM filters, compared to the PHD-SLAM filter~\cite{Hyowon_TWC2020}.
    In Fig.~\ref{Fig:VA_GOSPA} and Fig.~\ref{Fig:SP_GOSPA}, average GOSPAs for VA and SP are presented, respectively.
    We can see both VA and SP GOSPAs of the proposed three SLAM filters decrease over time steps while the GOSPAs of PHD fluctuate.
    Hence, the proposed SLAM filters are robust to both missed detections and false alarms, compared to~\cite{Hyowon_TWC2020}.
    Contrary to VAs, SPs have the large measurement covariance and limited FoV, resulting in generating multiple global hypotheses at each time step.
    {Among the three proposed filters, therefore, the SP GOSPA has significant gain in the PMBM filter that never goes through approximation to the global hypotheses~\cite[Sec.~III]{Jason_VAMB_TSP2014}, while the three filters exhibit similar VA GOSPA.}

\subsubsection{Vehicle Localization}
    Fig.~\ref{Fig:Positioning} shows the root mean square error~(RMSE) of the estimated vehicle location, bias, and heading.
    The PMBM, PMB, and PHD SLAM filters, implemented by the RBPF, have similar localization accuracy compared to the MPMB-SLAM filter.
    This is because that the noise covariance for the known BS is relatively smaller than VA and SP in the FIM-based noise covariance even there are differences in the VA and SP GOSPAs.
    {The localization performance of the MPMB filter is  worse than the other filters due to the Gaussian approximation to the posterior density with the nonlinear measurement and assumption that the prior vehicle and landmarks are independent in the joint CKF implementation. We believe that alternative linearizations can close this gap \cite{Angel_IPLF_TSP2015}.
    It results that the prior covariance of joint state between the vehicle and landmarks cannot capture the correlation that is captured in the update step with the measurement likelihood.}
    
    

\subsubsection{Computation Complexity}
    The computational complexity is represented as the operation time.
    During 40 time steps, the vehicle is traveling one lap along with the circular road.
    {The average computation time of the three proposed SLAM filters for per time steps is the following: (i) PMBM-SLAM: 4.7 minutes; (ii) PMB: 3.5 minutes; and (iii) marginalized PMB: 1.4 seconds.
    The operation time of PHD~\cite{Hyowon_TWC2020} is 1.8 minutes.}
    It is shown that PMB-SLAM yielded a reduction of the averaged operation time by 24.7~\% due to approximating the PMBM-SLAM as a PMB-SLAM by the marginal association probabilities and Bethe free energy.
    Furthermore, we confirm that the marginalized PMB-SLAM significantly yielded it by 99.5~\% due to marginalization using Lemma~\ref{lem:MinKLD} as well as PMB-SLAM approximation.
    Even the operation time of proposed PMBM and PMB-SLAM filters are about two times or more than PHD, it is adequate for using PMBM or PMB-SLAM filters, rather than PHD because both are robust to missed detections and false alarms.

\section{Conclusions}\label{sec:Conclusions}
    This paper tackled the SLAM problem using the PMBM filter. We have proposed, derived, and implemented three novel SLAM filters, based on the PMBM filter: the PMBM, PMB, and marginalized PMB-SLAM filters. 
    We showed that the three proposed filters are robust to both missed detections and false alarms in comparison with the PHD-SLAM filter, and confirmed that the computation-performance trade-off comes from the marginalization of auxiliary variables.
    The trade-off is especially prominent in the marginalized PMB-SLAM filter by marginalizing both auxiliary variables for global hypotheses and vehicle state. 
    {We also reveal close connections to the BP-SLAM filter, which opens the way to further complexity reductions and highly parallelized particle implementations to improve the accuracy of the marginalized PMB-SLAM filter. A more explicit connection requires additional study and is left for future work.} Other planned extensions of this work are to incorporate a theoretical performance analysis of the proposed filters, combination with a real channel estimator, and the generation of propagation paths from a ray tracing simulator; and to improve the localization accuracy of the marginalized PMB by the posterior linearization approximate with the iterated posterior linearization filter~\cite{Angel_IPLF_TSP2015} or by the parallelized particle filter~\cite[Chapter~5.3]{Henk_IRD_2007} without Gaussian approximation and measurement linearization.
    
    

\appendices

\begin{appendices}

\section{PMB Component and Global Hypothesis Update}
\label{sec:App-PMBM-update}
     
\subsection{PMB Component Update}\label{sec:App-PMBM-map-update}
    \noindent \paragraph*{Step i)} The intensity for $m\in \{\text{VA},\text{SP}\}$ is updated as $\lambda_{\mathsf{u},k}^n(\mathbf{x},m) = (1-\mathsf{p}_{\text{D},k}^n(\mathbf{x},m))\lambda_{\mathsf{u},k-1}^n(\mathbf{x},m),$
    where $\mathsf{p}_{\text{D},k}^n(\mathbf{x},m)$ is shorthand of $\mathsf{p}_\text{D}(\mathbf{s}_{k}^n,\mathbf{x},m)$, from Section \ref{sec:Models}.
    \paragraph*{Step ii)} The updated MBM components with  $\mathbf{z}_k^j\in\mathcal{Z}_k$ have density
\begin{align}
    f_{\mathsf{u},k}^{j,n}(\mathbf{x},m) & = \dfrac{\mathsf{p}_{\text{D},k}^n(\mathbf{x},m)\lambda_{\mathsf{u},k-1}^n(\mathbf{x},m)g(\mathbf{z}_k^j|\mathbf{s}_{0:k}^n,\mathbf{x},m)}{\sum_{m'} e_k^{j}(m')}
\end{align}
    for $m\in\{\text{BS},\text{VA},\text{SP}\}$, existence probability  $r_{\mathsf{u},k}^{j,n} = {\sum_m e_k^{j,n}(m)}/{\nu_k^{n}(\{\mathbf{z}_k^j\})}$, and weight
     $\beta_{\mathsf{u},k}^{j,n} = \nu_k^{n}(\{\mathbf{z}_k^j\})$,
    where 
    $ e_k^{j,n}(m)
    = \int \mathsf{p}_{\text{D},k}^n(\mathbf{x},m)\lambda_{\mathsf{u},k-1}^n(\mathbf{x},m)g(\mathbf{z}_k^j|\mathbf{s}_{0:k}^n,\mathbf{x},m)\mathrm{d}\mathbf{x}$
and $\nu_k^{n}(\{\mathbf{z}_k^j\}) = \sum_{m} e_k^{j,n}(m)+c(\mathbf{z}_k^j)$, 
    in which $c(\mathbf{z})$ is the clutter intensity. This Step ii) leads to new landmarks, which will be indexed $i
    \in \{I_{k-1}+1,\ldots, I_{k-1}+J_k\}$. 
    \paragraph*{Step iii)} Under missed detection, the MBM components are updated with density 
\begin{align}
    f_{\mathsf{u},k}^{0,i,a_{k-1}^i,n}(\mathbf{x},m) & = \dfrac{(1-\mathsf{p}_{\text{D},k}^n(\mathbf{x},m))f_{\mathsf{u},k-1}^{i,a_{k-1}^i,n}(\mathbf{x},m)}
    {\sum_{m'}e_{k}^{0,i,a_{k-1}^i,n}(m')},
\end{align}
    for landmark type $m$$\in$$\{\text{BS},\text{VA},\text{SP}\}$, existence probability $ r_{\mathsf{u},k}^{0,i,a_{k-1}^i,n}$$=$${r_{\mathsf{u},k-1}^{i,a_{k-1}^i,n}\sum_m e_{k}^{0,i,a_{k-1}^i,n}(m)}/{\nu_k^{i,a_{k-1}^i,n}(\emptyset)}$ and weight $ \beta_{\mathsf{u},k}^{0,i,a_{k-1}^i,n}$$=$$\beta_{\mathsf{u},k-1}^{i,a_{k-1}^i,n}\nu_k^{i,a_{k-1}^i,n}(\emptyset)$, 
    where $ e_{k}^{0,i,a_{k-1}^i,n}(m) $$=$$ \int(1-\mathsf{p}_{\text{D},k}^n(\mathbf{x},m))f_{\mathsf{u},k-1}^{i,a_{k-1}^i,n}(\mathbf{x},m) \mathrm{d} \mathbf{x}$ and $\nu_k^{i,a_{k-1}^i,n}(\emptyset) = 1-r_{\mathsf{u},k-1}^{i,a_{k-1}^i,n} + r_{\mathsf{u},k-1}^{i,a_{k-1}^i,n}\sum_m e_{k}^{0,i,a_{k-1}^i,n}(m)$.
\paragraph*{Step iv)}  The updated MBM components with  $\mathbf{z}_k^j\in\mathcal{Z}_k$ have density 
\begin{align}
    f_{\mathsf{u},k}^{j,i,a_{k-1}^i,n}(\mathbf{x},m)   
    &= \dfrac{\mathsf{p}_{\text{D},k}^n(\mathbf{x},m)f_{\mathsf{u},k-1}^{i,a_{k-1}^i,n}(\mathbf{x},m)g(\mathbf{z}_k^j|\mathbf{s}_{0:k}^n,\mathbf{x},m)}{\sum_{m'} e_k^{j,i,a_{k-1}^i,n}(m')},\nonumber
\end{align}
    for landmark type  $m$$\in$$\{\text{BS},\text{VA},\text{SP}\}$, existence probability $r_{\mathsf{u},k}^{j,i,a_{k-1}^i,n}$$=$$1$, and weight $\beta_{\mathsf{u},k}^{j,i,a_{k-1}^i,n}$$=$$\beta_{\mathsf{u},k-1}^{i,a_{k-1}^i,n}\nu_k^{i,a_{k-1}^i,n}(\{\mathbf{z}_k^j\})$,
    where $e_k^{j,i,a_{k-1}^i,n}(m) = \int \mathsf{p}_{\text{D},k}^n(\mathbf{x},m)f_{\mathsf{u},k-1}^{i,a_{k-1}^i,n}(\mathbf{x},m)g(\mathbf{z}_k^j|\mathbf{s}_{0:k}^n,\mathbf{x},m)\mathrm{d}\mathbf{x}$ and $ \nu_k^{i,a_{k-1}^i,n}(\{\mathbf{z}_k^j\})
    = r_{\mathsf{u},k-1}^{j,i,a_{k-1}^i,n}\sum_{m} e_k^{j,i,a_{k-1}^i,n}(m)$.

   \subsection{Global Hypothesis Update}\label{Sec:GHupdate}
    Using the updated MBM components above, the set of global hypotheses $\mathcal{A}_k^n$ is now updated \cite{williams2015marginal} by selecting the $B_k^{\mathbf{a}_{k-1},n}$ best global hypotheses using Murty's algorithm \cite{murthy1968algorithm,Garcia-Fernandez2018} for each $\mathbf{a}_{k-1}\in \mathcal{A}_{k-1}^n$.
    For each landmark $i$ under hypothesis $\mathbf{a}_{k-1}\in \mathcal{A}_{k-1}^n$, $J_k +1$
    local hypotheses are added (1 local hypothesis from Step iii) and $J_k$ local hypotheses from Step iv)); and for each measurement, a new local hypothesis (landmark or clutter) is created (Step ii)).
    {Finally, $\mathcal{A}_{k}^n$ comprises the hypotheses that are globally consistent
    (i.e., with at most 1 measurement for each landmark and at most 1 landmark associated to each measurement).}

\section{Proof of Particle Weight Updates}
     
\subsection{PMBM-SLAM} \label{app:PMBM-weigth-update}
    Finally, the particle weight $w_{\mathsf{u},k}^n$ is updated as 
\begin{align} \label{eq:WeightUpdate}
    &w_{\mathsf{u},k}^n \propto w_{\mathsf{p},k}^n g(\mathcal{Z}_k|\mathbf{s}_{0:k}^n,\mathcal{Z}_{1:k-1})\\
    & = w_{\mathsf{p},k}^n \int f(\mathcal{X}|\mathbf{s}_{0:k}^n,\mathcal{Z}_{1:k-1})g(\mathcal{Z}_k|\mathbf{s}_{0:k}^n,\mathcal{X},\mathcal{Z}_{1:k-1}) \delta\mathcal{X},
\end{align}
    which follows the RFS-likelihood in~\eqref{eq:BasicVehUp}.
    We plug the PMBM form of \eqref{eq:PMBdensity} into the predicted landmark density $f(\mathcal{X}|\mathbf{s}_{0:k}^n,\mathcal{Z}_{1:k-1})$, and adopt the likelihood representation in \cite[eqs.\,(25), (26)]{Garcia-Fernandez2018} for replacing $g(\mathcal{Z}_k|\mathbf{s}_{0:k}^n,\mathcal{X},\mathcal{Z}_{1:k-1})$. Therefore, we find that \cite[Sec. III-D]{Garcia-Fernandez2018}
\begin{align} 
    w_{\mathsf{u},k}^n
    \propto
    &\,w_{\mathsf{p},k}^n\int f(\mathcal{X}|\mathbf{s}_{0:k}^n,\mathcal{Z}_{1:k-1})g(\mathcal{Z}_k|\mathbf{s}_{0:k}^n,\mathcal{X},\mathcal{Z}_{1:k-1}) \delta\mathcal{X} \nonumber
\end{align}
    so that 
\begin{align}
   w_{\mathsf{u},k}^n \propto &\,w_{\mathsf{p},k}^n \sum_{\mathbf{a}_{k-1}\in \mathcal{A}_{k-1}}\sum_{ \substack{ \uplus_{i=1}^{I_{k-1}} \mathcal{Z}_k^i \uplus \mathcal{Z}_k^{\text{U}}= \mathcal{Z}_k,\\ |\mathcal{Z}_k^i|\le 1}}  \prod_{\mathbf{z}_k^j \in \mathcal{Z}_k^\text{U}}\nu_k^{n}({\{\mathbf{z}_k^j\}})\nonumber\\
    &\times \prod_{i=1}^{I_{k-1}} \nu_k^{i,a_{k-1}^i,n}(\mathcal{Z}_k^i)\beta_{\mathsf{u},k-1}^{i,a_{k-1}^i,n} = w_{\mathsf{p},k}^n \chi^n_k.
    \label{eq:ELComputeApp}
\end{align}

\subsection{PMB-SLAM} \label{app:PMB-weigth-update}

    To derive the weight computation of vehicle particle, we recover the normalization constant $\chi_k^n = Z_k^n$~(see,~\eqref{eq:ELCompute} and~\eqref{eq:PMBGlobalWeight}) by the Bethe free energy~\cite{Yedidia2005BetheFree}, defined as
    $- \ln Z_k^n \approx \mathsf{F}(\mathrm{bel}_k^n)  = \mathsf{U}(\mathrm{bel}_k^n)-\mathsf{H}(\mathrm{bel}_k^n)$, where $\mathsf{F}$ is the Bethe free energy, $\mathsf{U}$ denotes the average energy and $\mathsf{H}$ the entropy, which can be computed from the beliefs and the factor graph structure, according to \cite[eqs.~(37)--(38)]{Yedidia2005BetheFree}. 
    In the special case where the beliefs are nearly degenerate, i.e., $\mathrm{bel}_k^n(c_k^i) \approx \delta({c_k^i-\bar{c}_k^{i}})$ and $\mathrm{bel}_k^n(d_k^j) \approx \delta({d_k^j-\bar{d}_k^{j})}$, $\mathsf{H}(\mathrm{bel}_k^n) \approx 0$ and $ \mathsf{U}(\mathrm{bel}_k^n) \approx   -\sum_{i=1}^{I_{k-1}}\ln p_{\mathsf{a},k}^n(\bar{c}_k^{i}) - \sum_{j=1}^{J_{k}}\ln p_{\mathsf{a},k}^n(\bar{d}_k^{j})$, so that
\begin{align}
    Z_k^n & \approx \exp(-\mathsf{F}(\mathrm{bel}_k^n))
       \approx \prod_{i=1}^{I_{k-1}}\prod_{j=1}^{J_{k}}p_{\mathsf{a},k}^n(\bar{c}_k^{i})p_{\mathsf{a},k}^n(\bar{d}_k^{j}),\label{eq:BTFapp}
\end{align}
    which is straightforward to evaluate. 
\section{Conversion from PMBM to PMB} \label{app:conversion}
     
    {
    The marginal association probabilities are the track-oriented associations with the measurements, and each track $i$ represents one landmark with all possible landmark types.
    Let us denote the marginal association probabilities by $\mathsf{p}_k^{i,n}(j)$ for the previously detected landmarks $i\in \{1,...,I_{k-1}\}$ and for missed detections and detections $j\in \{0,...,J_k \}$; and $\mathsf{p}_k^{I_{k-1}+j,n}(0)$ for the newly detected landmarks or clutter $j\in \{0,...,J_k \}$. Here, $\mathsf{p}_k^{i,n}(j)$ is the marginal probability that previously detected landmark $i$ is associated to measurement $j$, and $\mathsf{p}_k^{I_{k-1}+j,n}(0)$ is the marginal probability that measurement $j$ corresponds to a newly detected landmark (the one with index $I_{k-1}+j$).
}
    We then apply the TOMB/P method from~\cite[Appendix C, Fig. 10]{williams2015marginal}: for existing landmarks $i = 1,...,I_{k-1}$
\begin{align}
    {r}_{\mathsf{u},k}^{i,n} &=\sum_{j=0}^{J_k} \mathsf{p}_k^{i,n}(j)r_{\mathsf{u},k}^{j,i,n}\label{eq:TOMBP_ELr},\\
    {f}_{\mathsf{u},k}^{i,n}(\mathbf{x},m) &= \frac{1}{{r}_{\mathsf{u},k}^{i,n}}\sum_{j=0}^{J_k} \mathsf{p}_k^{i,n}(j) r_{\mathsf{u},k}^{j,i,n}f_{\mathsf{u},k}^{j,i,n}(\mathbf{x},m)\label{eq:TOMBP_ELf},
\end{align}
    and for new landmarks (with $j = 1,...,J_k$)
\begin{align}
    {r}_{\mathsf{u},k}^{I_{k-1}+j,n} &= \mathsf{p}_k^{I_{k-1}+j,n}(0)r_{\mathsf{u},k}^{j,n},
    \label{eq:TOMBP_NLr}\\
    {f}_{\mathsf{u},k}^{I_{k-1}+j,n}(\mathbf{x},m) &= f_{\mathsf{u},k}^{j,n}(\mathbf{x},m).\label{eq:TOMBP_NLf}
\end{align}
     
\section{Derivation of Marginalized PMB Map Density}

\subsection{Mathematical Expressions}
\label{app:MPMB-veh}
    We will denote $h_{\mathsf{p},k}^\text{U}(\mathbf{s}_k,\mathbf{x},m) \triangleq f_{\mathsf{p},k}(\mathbf{s}_k)\lambda_{\mathsf{u},k-1}(\mathbf{x},m)$ and 
    $h_{\mathsf{p},k}^{\text{D},i}(\mathbf{s}_k,\mathbf{x},m) \triangleq f_{\mathsf{p},k}(\mathbf{s}_k)f_{\mathsf{u},k-1}^i(\mathbf{x},m)$.
 \begin{enumerate}[i),left=0pt,noitemsep,topsep=0.5pt]
     \item We update $\lambda_{\mathsf{u},k}(\mathbf{x},m)$ for $m\in\{\text{VA},\text{SP}\}$ as $\lambda_{\mathsf{u},k}(\mathbf{x},m) = (1-\int \mathsf{p}_{\text{D},k}(\mathbf{s}_k,\mathbf{x},m)\mathrm{d}\mathbf{s}_k)\lambda_{\mathsf{u},k-1}(\mathbf{x},m).$
     \item The MB components with $\mathbf{z}_k^j \in \mathcal{Z}_k$ are updated as
 \begin{align} \label{eq:MPMB-S2f}
     &f_{\mathsf{u},k}^{j}(\mathbf{x},m) \nonumber\\
     &= 
     \dfrac{\int  \mathsf{p}_{\text{D},k}(\mathbf{s}_k,\mathbf{x},m)h_{\mathsf{p},k}^\text{U}(\mathbf{s}_k,\mathbf{x},m)g(\mathbf{z}_k^j|\mathbf{s}_k,\mathbf{x},m) \mathrm{d}\mathbf{s}_k}
     {\sum_{m'}e_k^{j}(m')},
 \end{align}
     and $r_{\mathsf{u},k}^{j} = {\sum_m e_k^{j}(m)}/{\nu_k(\{\mathbf{z}_k^j\})}$. Here,
 \begin{align}
     e_k^{j}(m) = &\iint \mathsf{p}_{\text{D},k}(\mathbf{s}_k,\mathbf{x},m) h_{\mathsf{p},k}^\text{U}(\mathbf{s}_k,\mathbf{x},m) \nonumber\\
     & g(\mathbf{z}_k^j|\mathbf{s}_k,\mathbf{x},m) \mathrm{d}\mathbf{s}_k \mathrm{d}\mathbf{x}, \label{eq:MarUpC2_e}
 \end{align}
     and $\nu_k(\{\mathbf{z}_k^j\}) = \sum_{m} e_k^{j}(m)+c(\mathbf{z}_k^j)$.
     \item The MB components are updated as
\begin{align}\label{eq:MPMB-S3f}
      f_{\mathsf{u},k}^{0,i}(\mathbf{x},m) 
      =\dfrac{\int(1-\mathsf{p}_{\text{D},k}(\mathbf{s}_k,\mathbf{x},m))h_{\mathsf{p},k}^{\text{D},i}(\mathbf{s}_k,\mathbf{x},m)\mathrm{d}\mathbf{s}_k}
      {\sum_{m'} e_k^{0,i}(m')},
\end{align}
     and $r_{\mathsf{u},k}^{0,i} = {r_{\mathsf{u},k-1}^{i}\sum_m e_k^{0,i}(m)}/{\nu_k^{i}(\emptyset)}$.
     Here, 
\begin{align}
    &e_k^{0,i}(m) \nonumber\\
    &=  \iint(1-\mathsf{p}_{\text{D},k}(\mathbf{s}_k,\mathbf{x},m)) h_{\mathsf{p},k}^{\text{D},i}(\mathbf{s}_k,\mathbf{x},m) \mathrm{d} \mathbf{s}_k\mathrm{d}\mathbf{x},\label{eq:MarUpC3_e}
\end{align}
     and $\nu_k^{i}(\emptyset) = 1-r_{\mathsf{u},k-1}^{i} + r_{\mathsf{u},k-1}^{i}\sum_m e_k^{0,i}(m)$.
     \item The MB components with $\mathbf{z}_k^j \in \mathcal{Z}_k$ are updated as
\begin{align}
     &f_{\mathsf{u},k}^{j,i}(\mathbf{x},m) \nonumber\\
     &= \dfrac{\int \mathsf{p}_\text{D,k}(\mathbf{s}_k,\mathbf{x},m)
     h_{\mathsf{p},k}^{\text{D},i}(\mathbf{s}_k,\mathbf{x},m)
     g(\mathbf{z}_k^j|\mathbf{s}_k,\mathbf{x},m)
     \mathrm{d}\mathbf{s}_k}{\sum_{m'} e_k^{j,i}(m')},
\end{align}
     and $r_{\mathsf{u},k}^{j,i} = 1$.
     Here, $\nu_k^{i}(\{\mathbf{z}_k^j\})
     = r_{\mathsf{u},k-1}^{j,i}\sum_{m} e_k^{j,i}(m)$ and
 \begin{align}
     e_k^{j,i}(m) =& \iint  \mathsf{p}_\text{D,k}(\mathbf{s}_k,\mathbf{x},m)
     h_{\mathsf{p},k}^{\text{D},i}(\mathbf{s}_k,\mathbf{x},m)\nonumber\\
     &\times g(\mathbf{z}_k^j|\mathbf{s}_k,\mathbf{x},m) \mathrm{d}\mathbf{s}_k\mathrm{d}\mathbf{x}.\label{eq:MarUpC4_e}
 \end{align}
 \end{enumerate}
 
 \vspace{-7mm}
 \subsection{CKF Implementation}
\label{app:MPMB-mapCKF}

As derived in Appendix~\ref{app:MPMB-veh},
    the vehicle state is marginalized out in steps i)--iv). 
    We now implement these steps i)--iv) with CKF. 
    
\subsubsection{Steps i) and iii)}    
    Step i) and iii) correspond to the missed detections.
    Therefore, PPP and MB components can be computed without any intractable integration of the product of the prior density and likelihood function.
    

\subsubsection{Step ii)}
    We first show how to compute~\eqref{eq:MarUpC2_e} and then how to compute the spatial density~\eqref{eq:MPMB-S2f}.
\paragraph{Normalization constant for newly detected landmarks~\eqref{eq:MarUpC2_e}}
    We recap~\eqref{eq:MarUpC2_e} and calculate $e_k^{j}(m)$ as
\begin{align}\label{eq:MarUpC2_eCom}
    &e_k^{j}(m) \notag\\
    & = \iint \mathsf{p}_{\text{D},k}(\mathbf{s}_k,\mathbf{x},m) h_{\mathsf{p},k}^\text{U}(\mathbf{s}_k,\mathbf{x},m)  g(\mathbf{z}_k^j|\mathbf{s}_k,\mathbf{x},m) \mathrm{d}\mathbf{s}_k \mathrm{d}\mathbf{x} \notag \\
    & = 
    P_\text{D} 
    \kappa_{\mathsf{u},k-1}(m)
    \iint \mathcal{N}(\mathbf{s}_k;\mathbf{s}_{\mathsf{p},k},\mathbf{U}_{\mathsf{p},k})
    \mathcal{U}(\mathbf{x})\nonumber\\
    &\times \mathcal{N}(\mathbf{z}_k^j;\mathsf{h}(\mathbf{s}_{k},\mathbf{x},m),{\mathbf{R}_k^j})  \mathrm{d}\mathbf{s}_k \mathrm{d}\mathbf{x}.
\end{align}
    We draw $B$ independent and identically distributed (iid) samples $\mathbf{x}_k^{b,j}(m)\sim q_k^j(\mathbf{x},m)=\mathcal{N}(\mathbf{x};\bar{\mathbf{x}}_k^j(m),\bar{\mathbf{P}}_k^j(m))$ with  weights $w_k^{b,j} \propto {\mathcal{U}(\mathbf{x})}/{q(\mathbf{x}_k^{b,j}(m),m)}$, and $\sum_b w_k^{b,j} = 1$. 
    We select the proposal distribution based on the measurement, as follows. We set
    $\bar{\mathbf{x}}_{k}^j(m)=\sum_{c=1}^{2\mathsf{d}(\mathbf{s}_k)}\mathbf{x}_k^{c,j}(m)/{2\mathsf{d}(\mathbf{s}_k)}$ and $\bar{\mathbf{P}}_{k}^j(m)=\sum_{c=1}^{2\mathsf{d}(\mathbf{s}_k)}\big[  \mathbf{P}_{k}^{c,j}(m)+(\mathbf{x}_k^{c,j}(m)-\bar{\mathbf{x}}_{k}^j(m))(\mathbf{x}_k^{c,j}(m)-\bar{\mathbf{x}}_{k}^j(m))^\top]/{2\mathsf{d}(\mathbf{s}_k)}$,
    where $\mathbf{x}_{k}^{c,j}(m)=\sum_{b=1}^{2\mathsf{d}(\mathbf{z}_k)}\mathbf{x}_k^{b,c,j}(m)/{2\mathsf{d}(\mathbf{z}_k)}$ and $\mathbf{P}_{k}^{c,j}(m)=\sum_{b=1}^{2\mathsf{d}(\mathbf{z}_k)} (\mathbf{x}_k^{b,c,j}(m)-\mathbf{x}_{k}^{c,j}(m))
    (\mathbf{x}_k^{b,c,j}(m)-\mathbf{x}_{k}^{c,j}(m))^\top/{2\mathsf{d}(\mathbf{z}_k)}$.
    Here, $\mathbf{x}_k^{b,c,j}(m)$ is the birth point, computed as
\begin{align}\label{eq:VA_birth}
    &\mathbf{x}_k^{b,c,j}(\textrm{VA}) =
    \begin{bmatrix}
        x_{\mathsf{p},k}^c + r_k^{b,j}\cos(\theta_{\textrm{az},k}^{b,j} + \alpha_{\mathsf{p},k}^c)\\
        y_{\mathsf{p},k}^c + r_k^{b,j}\sin(\theta_{\textrm{az},k}^{b,j} + \alpha_{\mathsf{p},k}^c)\\
        z_{\mathsf{p},k}^c + \tau_k^{b,j}\sin(\theta_{\textrm{el},k}^{b,j})
    \end{bmatrix}
\end{align}
\begin{align}\label{eq:SP_birth}
    &\mathbf{x}_k^{b,c,j}(\textrm{SP}) \\
    &= \mathbf{x}_k^{b,c,j}(\textrm{VA}) + \frac{(\mathbf{f}_k-\mathbf{x}_k^{b,c,j}(\textrm{VA}))^\top\mathbf{u}_k}
    {(\mathbf{x}_{\mathsf{p},k}^c-\mathbf{x}_k^{b,c,j}(\textrm{VA}))^\top \mathbf{u}_k}(\mathbf{x}_{\mathsf{p},k}^c-\mathbf{x}_k^{b,c,j}(\textrm{VA})), \nonumber
\end{align}
    where $r_k^{b,j} = (\tau_k^{b,j}-B_k^{b,j})\cos(\theta_{\textrm{el},k}^{b,j})$, $\mathbf{u}_k = ({\mathbf{x}_{\text{BS}} - \mathbf{x}_k^{b,c,j}(\textrm{VA})})/{\Vert \mathbf{x}_{\text{BS}} - \mathbf{x}_k^{b,c,j}(\textrm{VA}) \Vert}$, and $\mathbf{f}_k = 0.5 \times ({\mathbf{x}_{\text{BS}} + \mathbf{x}_k^{b,c,j}(\textrm{VA})})$.
    Here,  $\mathbf{s}_{\mathsf{p},k}^c$ and $\mathbf{z}_{k}^{b,j}$ are the CPs with the cubature index $c$ and $b$, respectively, which are generated from $\mathcal{N}(\mathbf{s}_k;\mathbf{s}_{\mathsf{p},k},\mathbf{U}_{\mathsf{p},k})$ and $\mathcal{N}(\mathbf{z}_k;\mathbf{z}_k^j,{\mathbf{R}_k^j})$.
    Putting this together, then we have
\begin{align}
    e_k^{j}(m) \approx & \frac{P_\text{D}}{B}
    \kappa_{\mathsf{u},k-1}(m)
    \sum_b w_k^{b,j}(m)\int \mathcal{N}(\mathbf{s}_k;\mathbf{s}_{\mathsf{p},k},\mathbf{U}_{\mathsf{p},k}) \nonumber\\
    &\times
    \mathcal{N}(\mathbf{z}_k^j;\mathsf{h}(\mathbf{s}_{k},\mathbf{x}_k^{b,j}(m),m),{\mathbf{R}_k^j})  \mathrm{d}\mathbf{s}_k,
\end{align}
    where $\mathbf{x}_k^{b,j}(m)= \sum_{c=1}^{2\mathsf{d}(\mathbf{s}_k)}\mathbf{x}_k^{b,c,j}(m)/{2\mathsf{d}(\mathbf{s}_k)}$.
    By CKF approximation, we find
\begin{align}
   e_k^{j}(m) \approx & \frac{P_\text{D}}{B}
    \kappa_{\mathsf{u},k-1}(m)
    \sum_b w_k^{b,j}(m)
    \nonumber\\
    &\times
    \mathcal{N}(\mathbf{z}_k^j;\mathsf{h}(\mathbf{s}_{\mathsf{p},k},\mathbf{x}_k^{b,j}(m),m),\mathbf{S}_{\text{zz},k}^{b,j}(m)).
\end{align}

\paragraph{Landmark density for newly detected landmarks~\eqref{eq:MPMB-S2f}}
    {With the CPs $\mathbf{x}_k^{b,c,j}(m)$ of ~\eqref{eq:VA_birth} and~\eqref{eq:SP_birth}, the landmark density of~\eqref{eq:MPMB-S2f} is computed as $f_{\mathsf{u},k}^j(\mathbf{x},m)=\mathcal{N}(\mathbf{x};\bar{\mathbf{x}}_k^j(m),\bar{\mathbf{P}}_k^j(m))$, where $\bar{\mathbf{x}}_k^j(m)$ and $\bar{\mathbf{P}}_k^j(m)$ were handled in~\eqref{eq:MarUpC2_eCom}.}
    

\subsubsection{Step iv)}
    Similar to Step ii), we first compute the normalization constant~\eqref{eq:MarUpC4_e} and then the spatial density~\eqref{eq:MPMB-S3f}. 


    
\paragraph{Normalization constant for previously detected landmarks \eqref{eq:MarUpC4_e}}
    To implement~\eqref{eq:MarUpC4_e}, we recap the expression
\begin{align}
     & e_k^{j,i}(m) \notag\\
     & = \iint  \mathsf{p}_\text{D,k}(\mathbf{s}_k,\mathbf{x},m)
     h_{\mathsf{p},k}^{\text{D},i}(\mathbf{s}_k,\mathbf{x},m) g(\mathbf{z}_k^j|\mathbf{s}_k,\mathbf{x},m) \mathrm{d}\mathbf{s}_k\mathrm{d}\mathbf{x}.\notag
 \end{align}
    To solve the integral, we construct $\mathcal{N}(\mathbf{j}_k;\mathbf{j}_{\mathsf{p},k}^i(m),\mathbf{J}_{\mathsf{p},k}^i(m))$, where $\mathbf{j}_{\mathsf{p},k}^i(m)=[\mathbf{s}_{\mathsf{p},k}^\top, (\mathbf{x}_{\mathsf{p},k}^{i}(m))^\top]^\top$ and $\mathbf{J}_{\mathsf{p},k}^i(m)=\text{blkdiag}(\mathbf{U}_{\mathsf{p},k}, \mathbf{P}_{\mathsf{p},k}^{i}(m))$.
    We approximate $e_k^{j,i}(m)$ as
\begin{align}
    e_k^{j,i}(m) 
    \approx &
    \iint p_{\mathsf{p},k}^{i}(m)p_\text{AD}(m) \mathcal{N}(\mathbf{j}_k;\mathbf{j}_{\mathsf{p},k}^{i}(m),\mathbf{J}_{\mathsf{p},k}^{i}(m))
    \nonumber\\
    &\times \mathcal{N}(\mathbf{z}_k^j;\mathsf{h}(\mathbf{s}_{k},\mathbf{x},m),{\mathbf{R}_k^j})\mathrm{d}\mathbf{s}_k\mathrm{d}\mathbf{x}.
\end{align}
    Using the CKF~\cite{HaykinCKF2009} approximation
\begin{align}
    \label{eq:iv_CKFapprox}
    &\mathcal{N}(\mathbf{j}_k;\mathbf{j}_{\mathsf{p},k}^{i}(m),\mathbf{J}_{\mathsf{p},k}^{i}(m))\mathcal{N}(\mathbf{z}_k^j;\mathsf{h}(\mathbf{s}_{k},\mathbf{x},m),{\mathbf{R}_k^j})\\
    &\approx \mathcal{N}(\mathbf{j};\mathbf{j}_{\mathsf{u},k}^{j,i}(m),\mathbf{J}_{\mathsf{u},k}^{j,i}(m))\mathcal{N}(\mathbf{z}_k^j;\mathsf{h}(\mathbf{s}_{\mathsf{p},k},\mathbf{x}_{\mathsf{p},k}^i(m)),\mathbf{P}_{\text{zz}}^{j,i}(m)),\nonumber
\end{align}
    $e_k^{j,i}(m) 
    \approx p_{\mathsf{p},k}^{i}(m) p_\text{AD}(m) \mathcal{N}(\mathbf{z}_k^j;\mathsf{h}(\mathbf{s}_{\mathsf{p},k},\mathbf{x}_{\mathsf{p},k}^i,m),\mathbf{P}_{\text{zz}}^{j,i}(m)).$
\paragraph{Landmark density for previously detected landmarks~\eqref{eq:MPMB-S3f}}
    {The landmark density of~\eqref{eq:MPMB-S3f} is computed as $f_{\mathsf{u},k}^{0,i}(\mathbf{x},m)=\mathcal{N}(\mathbf{x};\mathbf{x}_{\mathsf{u},k}^{j,i}(m),\mathbf{P}_{\mathsf{u},k}^{j,i}(m))$.
    We extract $\mathbf{x}_{\mathsf{u},k}^{j,i}(m)$ and $\mathbf{P}_{\mathsf{u},k}^{j,i}(m)$
    from $\mathcal{N}(\mathbf{j};\mathbf{j}_{\mathsf{u},k}^{j,i}(m),\mathbf{J}_{\mathsf{u},k}^{j,i}(m))$ of~\eqref{eq:iv_CKFapprox}, where $\mathbf{j}_{\mathsf{u},k}^{j,i}(m)=[\bar{\mathbf{s}}_{\mathsf{u},k}^\top, (\mathbf{x}_{\mathsf{u},k}^{i}(m))^\top]^\top$ and $\mathbf{J}_{\mathsf{u},k}^{j,i}(m) = [\mathbf{U}_{\mathsf{u},k}, \mathbf{O}_{\mathsf{u},k}; \mathbf{O}_{\mathsf{u},k}^\top, \mathbf{P}_{\mathsf{u},k}^{i}(m)]$.}
\section{Derivation of Marginalized Vehicle Posterior}
\label{app:MPMB-target}

    Based on~\cite[eq.~(33)]{Garcia-Fernandez2018}, we have
\begin{align}
    f_{\mathsf{u},k}(\mathbf{s}_k)  \propto & \int f_{\mathsf{p},k}(\mathbf{s}_k) \sum_{\uplus_{i=1}^{I_{k-1}}\mathcal{X}^i \uplus \mathcal{X}^\text{U} \notag =\mathcal{X}}\sum_{\uplus_{i=1}^{I_{k-1}}\mathcal{Z}_k^i \uplus \mathcal{Z}_k^\text{U} =\mathcal{Z}_k} \\ 
    & \times 
    f^\text{U}_{\mathsf{u},k-1}(\mathcal{X}^\text{U})
    l(\mathcal{Z}_k^\text{U}|\mathbf{s}_k,\mathcal{X}^\text{U}) \notag
    \\
    &\times
    \prod_{i=1}^{I_{k-1}}
    f^i_{\mathsf{u},k-1}(\mathcal{X}^i)
    t(\mathcal{Z}_k^i|\mathbf{s}_k,\mathcal{X}^i) \delta \mathcal{X},
\end{align}
    where $\lvert \mathcal{X}^i \rvert \leq 1$. We denote $\mathcal{Z}_k^\text{U}=\{\mathbf{z}_k^1,...,\mathbf{z}_k^{\lvert \mathcal{Z}_k^\text{U} \rvert}\}$, and we decompose $\mathcal{X}^\text{U}$ into all possible sets $\mathcal{U},\mathcal{Y}^1,...,\mathcal{Y}^{\lvert \mathcal{Z}_k^\text{U} \rvert}$, where $\mathcal{U}$ is a set of undetected landmarks, and set $\mathcal{Y}$ is the origin of the measurement $\mathbf{z} \in \mathcal{Z}_k^\text{U}$. Then, $l(\mathcal{Z}_k^\text{U}|\mathbf{s}_k,\mathcal{X}^\text{U})$ is given by \cite[eq. (13)]{Garcia-Fernandez2018}
\begin{align}\label{eq:likelihood_UD}
    &l(\mathcal{Z}_k^\text{U}|\mathbf{s}_k,\mathcal{X}^\text{U}) \nonumber\\
    &=  e^{-\int c(\mathbf{z})\mathrm{d}\mathbf{z}} \sum_{\uplus_{i=1}^{\lvert \mathcal{Z}_k^\text{U} \rvert} \mathcal{Y}^i \uplus \mathcal{U}  =\mathcal{X}^\text{U}} \prod_{(\mathbf{x},m)\in \mathcal{U}} [1-\mathsf{p}_\text{D}(\mathbf{s}_k,\mathbf{x},m)] \nonumber\\ &~~ \times \prod_{\mathbf{z} \in \mathcal{Z}_k^\text{U}} \tilde{l}(\mathbf{z}|\mathbf{s}_k,\mathcal{Y}^i),
\end{align}
    where  $\tilde{l}(\mathbf{z}|\mathbf{s}_k,\mathcal{Y})$ is given by
\begin{align}\label{eq:likelihood_Poi}
    &\tilde{l}(\mathbf{z}|\mathbf{s}_k,\mathcal{Y}) =
    \begin{cases}
        \mathsf{p}_\text{D}(\mathbf{s}_k,\mathbf{x},m)g(\mathbf{z}|\mathbf{s}_k,\mathbf{x},m) & \mathcal{Y} = \{(\mathbf{x},m)\}, \\
        c(\mathbf{z}) & \mathcal{Y} = \emptyset,\\
        0 & \lvert \mathcal{Y} \rvert > 1,
    \end{cases}
\end{align}
    and $t(\mathcal{Z}_k^i|\mathbf{s}_k,\mathcal{X}^i)$ is given by \cite[eq. (26)]{Garcia-Fernandez2018}
\begin{align}\label{eq:likelihoodBer}
    &t(\mathcal{Z}_k^i|\mathbf{s}_k,\mathcal{X}^i)\nonumber \\
    &=
    \begin{cases}
        \mathsf{p}_\text{D}(\mathbf{s}_k,\mathbf{x},m)g(\mathbf{z}|\mathbf{s}_k,\mathbf{x},m) & \mathcal{Z}_k^i=\{\mathbf{z} \},~\mathcal{X}^i=\{(\mathbf{x},m)\},\\
        1-\mathsf{p}_\text{D}(\mathbf{s}_k,\mathbf{x},m) & \mathcal{Z}_k^i = \emptyset,~\mathcal{X}^i=\{(\mathbf{x},m)\},\\
        1 & \mathcal{Z}_k^i = \emptyset,~\mathcal{X}^i=\emptyset,\\
        0 & \text{otherwise}.
    \end{cases}
\end{align}
    Making use of~\cite[Corollary~2]{Garcia-Fernandez2018}\cite[Lemma~2]{Jason_VAMB_TSP2014}, which states that
    $ \int \sum_{\mathcal{X} \uplus \mathcal{Y}=\mathcal{Z}}f(\mathcal{X})g(\mathcal{Y}) \delta \mathcal{Z} = \int f(\mathcal{X}) \delta\mathcal{X} \int g(\mathcal{Y}) \delta\mathcal{Y}$, 
    we find
\begin{align}
    &f_{\mathsf{u},k}(\mathbf{s}_k) \label{eq:MarVehicleUp}\\
    &\propto  \sum_{\uplus_{i=1}^{I_{k-1}}\mathcal{Z}_k^i \uplus \mathcal{Z}_k^\text{U} =\mathcal{Z}_k} f_{\mathsf{p},k}(\mathbf{s}_k)q(\mathcal{Z}_k^\text{U}|\mathbf{s}_k)\prod_{i=1}^{I_{k-1}} q(\mathcal{Z}_k^i|\mathbf{s}_k), \nonumber
\end{align}
    where $q(\mathcal{Z}_k^\text{U}|\mathbf{s}_k)$ and $q(\mathcal{Z}_k^\text{U}|\mathbf{s}_k)$ are
\begin{align}
    q(\mathcal{Z}_k^\text{U}|\mathbf{s}_k) 
    &= \int f^\text{U}_{\mathsf{u},k-1}(\mathcal{X}^\text{U})l(\mathcal{Z}_k^\text{U}|\mathbf{s}_k,\mathcal{X}^\text{U}) \delta\mathcal{X}^\text{U},\label{eq:likelihood_zs2}\\
    q(\mathcal{Z}_k^i|\mathbf{s}_k) &= \int f^i_{\mathsf{u},k-1}(\mathcal{X}^i)t(\mathcal{Z}_k^i|\mathbf{s}_k,\mathcal{X}^i)\delta\mathcal{X}^i.\label{eq:likelihood_zsBer}
\end{align}
    Substituting~\eqref{eq:likelihood_UD} into~\eqref{eq:likelihood_zs2} and invoking again~\cite[Corollary~2]{Garcia-Fernandez2018}, 
    we find that $q (\mathcal{Z}_k^\text{U}|\mathbf{s}_k)$ is a PPP since 
\begin{align}
     q (\mathcal{Z}_k^\text{U}|\mathbf{s}_k)
    =&  e^{-\int c(\mathbf{z})\mathrm{d}\mathbf{z}} e^{-\sum_m\int \lambda_{\mathsf{u},k-1}(\mathbf{x},m) \mathrm{d}\mathbf{x} } \notag \\
    & \times\int \prod_{(\mathbf{x},m)\in \mathcal{U}}\lambda_{\mathsf{u},k-1}(\mathbf{x},m)(1-\mathsf{p}_\text{D}(\mathbf{s}_k,\mathbf{x},m)) \delta \mathcal{U} \nonumber\\
    & \times \prod_{\mathbf{z} \in \mathcal{Z}_k^\text{U}}
    \int \prod_{(\mathbf{x},m)\in \mathcal{Y}}
    \lambda_{\mathsf{u},k-1}
    (\mathbf{x},m)
    l(\mathbf{z}|\mathbf{s}_k,\mathcal{Y}) \delta\mathcal{Y} \label{eq:likelihood_zs1}\\
    \propto &  e^{-\int \psi_k(\mathbf{z},\mathbf{s}_k)\mathrm{d}\mathbf{z}} \prod_{\mathbf{z} \in \mathcal{Z}_k^\text{U}} \psi_k(\mathbf{z},\mathbf{s}_k)
\end{align}
with intensity function
\begin{align}
    & \psi_k(\mathbf{z},\mathbf{s}_k) \label{eq:appIntensity}
    \\&=c(\mathbf{z}) +\sum_m\int\mathsf{p}_\text{D}(\mathbf{s}_k,\mathbf{x},m)\lambda_{\mathsf{u},k-1}(\mathbf{x},m)   g(\mathbf{z}|\mathbf{s}_k,\mathbf{x},m)\mathrm{d}\mathbf{x}. \notag
\end{align}
    Substituting~\eqref{eq:likelihoodBer} into~\eqref{eq:likelihood_zsBer} with the data association, we find that if $\mathcal{Z}_k^i=\{\mathbf{z} \}$,
\begin{align}
    q_i(\mathcal{Z}_k^i|\mathbf{s}_k) = & r_{\mathsf{u},k-1}^i \sum_m \int\mathsf{p}_\text{D}(\mathbf{s}_k,\mathbf{x},m) f_{\mathsf{u},k-1}^i(\mathbf{x},m) \nonumber\\
    &\times g(\mathbf{z}|\mathbf{s}_k,\mathbf{x},m) \mathrm{d}\mathbf{x}, \label{eq:app-qZi-detected}
\end{align}
    and that if $\mathcal{Z}_k^i=\emptyset$,
\begin{align}
    q_i(\mathcal{Z}_k^i|\mathbf{s}_k)=&r_{\mathsf{u},k-1}^i \sum_m \int(1-\mathsf{p}_\text{D}(\mathbf{s}_k,\mathbf{x},m))f_{\mathsf{u},k-1}^i(\mathbf{x},m)\mathrm{d}\mathbf{x}\nonumber\\
    &+ 1-r_{\mathsf{u},k-1}^i. \label{eq:app-qZi-misdetected}
\end{align}
    Using the sense of KLD minimization of Lemma~\ref{lem:MinKLD},
    we approximate the normalization constant for each global hypothesis in~\eqref{eq:MarVehicleUp} as
\begin{align}
    &\int f_{\mathsf{p},k}(\mathbf{s}_k)  q(\mathcal{Z}_k^\text{U}|\mathbf{s}_k)\prod_{i=1}^{I_{k-1}} q_i(\mathcal{Z}_k^i|\mathbf{s}_k)\mathrm{d}\mathbf{s}_k \label{eq:normalconstveh} \\
    &\approx\int f_{\mathsf{p},k}(\mathbf{s}_k)  q(\mathcal{Z}_k^\text{U}|\mathbf{s}_k)\mathrm{d}\mathbf{s}_k \prod_{i=1}^{I_{k-1}} \int f_{\mathsf{p},k}(\mathbf{s}_k) q_i(\mathcal{Z}_k^i|\mathbf{s}_k)\mathrm{d}\mathbf{s}_k. \notag
\end{align}
Using again the sense of KLD minimization similarly to~\eqref{eq:normalconstveh}, we compute the first factor of~\eqref{eq:normalconstveh}:
\begin{align}
    &\int f_{\mathsf{p},k}(\mathbf{s}_k)q(\mathcal{Z}_k^\text{U}|\mathbf{s}_k)\mathrm{d}\mathbf{s}_k 
    \appropto \int f_{\mathsf{p},k}(\mathbf{s}_k)e^{-\int \psi_k(\mathbf{z},\mathbf{s}_k)\mathrm{d}\mathbf{z}}\mathrm{d}\mathbf{s}_k \nonumber\\
    &  \times \prod_{\mathbf{z} \in \mathcal{Z}_k^\text{U}} 
    \int f_{\mathsf{p},k}(\mathbf{s}_k)\psi_k(\mathbf{z},\mathbf{s}_k)\mathrm{d}\mathbf{s}_k \\
     & \propto \prod_{\mathbf{z} \in \mathcal{Z}_k^\text{U}} (c(\mathbf{z})+\sum_m e_k(m)) =  \prod_{\mathbf{z} \in \mathcal{Z}_k^\text{U}} \nu_k(\{\mathbf{z}\})
\end{align}
where the constant $\int f_{\mathsf{p},k}(\mathbf{s}_k)e^{-\int \psi_k(\mathbf{z},\mathbf{s}_k)\mathrm{d}\mathbf{z}}\mathrm{d}\mathbf{s}_k$ 
{is} identical for all terms in \eqref{eq:MarVehicleUp} and 
\begin{align}
& e_k(m)=\\
& \iint f_{\mathsf{p},k}(\mathbf{s}_k)
    \mathsf{p}_\text{D}(\mathbf{s}_k,\mathbf{x},m)\lambda_{\mathsf{u},k-1}(\mathbf{x},m)   g(\mathbf{z}|\mathbf{s}_k,\mathbf{x},m)\mathrm{d}\mathbf{x}
    \mathrm{d}\mathbf{s}_k \notag
\end{align}
is identical to \eqref{eq:MarUpC2_e} in Appendix~\ref{app:MPMB-target}.

We compute the second factor of \eqref{eq:normalconstveh} as
\begin{align}\label{eq:intsecterm}
    \int f_{\mathsf{p},k}(\mathbf{s}_k)q_i(\mathcal{Z}_k^i|\mathbf{s}_k)\mathrm{d}\mathbf{s}_k =\nu_k^i(\mathcal{Z}_k^i),
\end{align}
    where $\nu_k(\{\mathbf{z}\})$, $\nu_k^i(\emptyset)$, and $\nu_k^i(\{\mathbf{z}\})$ were determined in Appendix~\ref{app:MPMB-veh}.

    We then find \eqref{eq:MarVehPosApp} by substitution of the normalized densities 
    of \eqref{eq:likelihood_zs2} and \eqref{eq:likelihood_zsBer} into \eqref{eq:MarVehicleUp}:
\begin{align}
    f_{\mathsf{u},k}(\mathbf{s}_k)  = & \sum_{\uplus_{i=1}^{I_{k-1}}\mathcal{Z}_k^i \uplus \mathcal{Z}_k^\text{U} =\mathcal{Z}_k}
    \prod_{\mathbf{z} \in \mathcal{Z}_k^\text{U}} \nu_k(\{\mathbf{z}\}) \prod_{i=1}^{I_{k-1}} \nu_k^i(\mathcal{Z}_k^i)\\
     & \times q(\mathbf{s}_k|\mathcal{Z}_k^\text{U},\mathcal{Z}_k^1,\ldots, \mathcal{Z}_k^{I_{k-1}})
     \nonumber
\end{align}
    where $q(\mathbf{s}_k|\mathcal{Z}_k^\text{U},\mathcal{Z}_k^1,\ldots, \mathcal{Z}_k^{I_{k-1}})$ is a normalized density, proportional to $ f_{\mathsf{p},k}(\mathbf{s}_k)  q(\mathcal{Z}_k^\text{U}|\mathbf{s}_k)\prod_{i=1}^{I_{k-1}} q_i(\mathcal{Z}_k^i|\mathbf{s}_k)$. 
    {To recover \eqref{eq:MarVehPosApp}, we note that the summation $\sum_{\uplus_{i=1}^{I_{k-1}}\mathcal{Z}_k^i \uplus \mathcal{Z}_k^\text{U} =\mathcal{Z}_k}$ is equivalent to the summation over global hypotheses $\mathbf{a}_k$ and that the weights of global hypotheses are given by \eqref{eq:PMBGlobalWeight}:
\begin{align}
    f_{\mathsf{u},k}(\mathbf{s}_k)  & =  \sum_{\mathbf{a}_k} \mathsf{p}(\mathbf{a}_k) q(\mathbf{s}_k|\mathbf{a}_k)\\
    & = \sum_{\mathbf{a}_k} \mathsf{p}(\mathbf{a}_k)  \frac{\Phi(\mathbf{s}_k|\mathbf{a}_k)}
    {\int \Phi(\mathbf{s}'_k|\mathbf{a}_k)\mathrm{d}\mathbf{s}'_k},
\end{align}
    where $\Phi(\mathbf{s}_k|\mathbf{a}_k)$ is given by
\begin{align}
    \Phi(\mathbf{s}_k|\mathbf{a}_k) = f_{\mathsf{p},k}(\mathbf{s}_k)\prod_{j \in \text{U}_k(\mathbf{a}_k)} \psi_k(\mathbf{z}^j,\mathbf{s}_k) \prod_{i=1}^{I_{k-1}} q_i(\mathcal{Z}_k^{a^i_k}|\mathbf{s}_k).
\end{align}
    Finally, we find that
\begin{align}
    & f_{\mathsf{u},k}(\mathbf{s}_k) \\
    & \approx \sum_{\mathbf{a}_k} \mathsf{p}(\mathbf{a}_k) f_{\mathsf{p},k}(\mathbf{s}_k) 
    \prod_{j \in \text{U}_k(\mathbf{a}_k)} \frac{\psi_k(\mathbf{z}^j,\mathbf{s}_k)}{\nu_k(\{\mathbf{z}^{j}\})} \prod_{i=1}^{I_{k-1}}
    \frac{q_i(\mathcal{Z}_k^{a^i_k}|\mathbf{s}_k)}{\nu_k^{i}(\mathcal{Z}_k^{a^{i}_k})}.\notag
\end{align}}

\end{appendices}
\bibliographystyle{IEEEtran}
\bibliography{bibliography}

\end{document}